\documentclass[twocolumn,a4paper,10pt]{article} %{bmcart}
\usepackage[margin=1in]{geometry}

\usepackage{titlesec}

\titleformat*{\section}{\large\bfseries}
\titleformat*{\subsection}{\normalsize\sffamily}
\titleformat*{\subsubsection}{\normalsize\itshape}

\usepackage{amsthm,amsmath,amssymb}
\usepackage{mathtools}
\usepackage{mathrsfs}
\usepackage{graphicx}
\usepackage{algorithm}
\usepackage{algorithmic}
\usepackage[version=4]{mhchem}
\usepackage{url}
\usepackage{xspace}
\usepackage{xcolor}

\makeatletter
\newcommand{\cupdot}{\charfusion[\mathbin]{\cup}{\cdot}}

\makeatletter
\def\moverlay{\mathpalette\mov@rlay}
\def\mov@rlay#1#2{\leavevmode\vtop{%
   \baselineskip\z@skip \lineskiplimit-\maxdimen
   \ialign{\hfil$\m@th#1##$\hfil\cr#2\crcr}}}
\newcommand{\charfusion}[3][\mathord]{
    #1{\ifx#1\mathop\vphantom{#2}\fi
        \mathpalette\mov@rlay{#2\cr#3}
      }
    \ifx#1\mathop\expandafter\displaylimits\fi}
\DeclareRobustCommand\bigop[1]{%
  \mathop{\vphantom{\sum}\mathpalette\bigop@{#1}}\slimits@
}
\newcommand{\bigop@}[2]{%
  \vcenter{%
    \strbalox\z@{$#1\sum$}%
    \hbox{\resizebox{\ifx#1\displaystyle.9\fi\dimexpr\ht\z@+\dp\z@}{!}{$\m@th#2$}}%
  }%
}
\makeatother

\newcommand{\trans}{\top}
\newcommand{\N}{\mathbb{N}}
\newcommand{\R}{\mathbb{R}}
\newcommand{\RR}{\mathscr{R}}
\newcommand{\SM}{\mathbf{S}} % stoich. matrix
\newcommand{\vv}{\mathbf{v}} % vector v
\newcommand{\gf}{\mathbf{g}} % (Gibbs free) reaction energy
\newcommand{\GF}{\mathbf{G}} % (Gibbs free) moluecular energy
\newcommand{\mm}{\mathbf{m}} 
\newcommand{\yy}{\mathbf{y}} 
\newcommand{\chemlike}{chemistry-like\xspace}
\DeclareMathOperator{\val}{val} 
\DeclareMathOperator{\im}{im} % like \ker
\DeclareMathOperator{\cone}{cone} 
\DeclareMathOperator{\supp}{supp}
\DeclareMathOperator{\lcd}{lcd}
\definecolor{sepia}{rgb}{0.44, 0.26, 0.08}
\definecolor{rotig}{rgb}{0.78, 0.12, 0.06}
\definecolor{darkgreen}{rgb}{0.2, 0.6, 0.2}
\definecolor{applegreen}{rgb}{0.55, 0.71, 0.0}
\definecolor{amber}{rgb}{1.0, 0.75, 0.0}
\definecolor{pfscolor}{rgb}{0.4, 0.5, 0.9}
\definecolor{purple}{RGB}{180,90,200}
\definecolor{dgreen}{RGB}{0,160,0}
\definecolor{turquoise}{RGB}{0,180,140}

\definecolor{revision}{rgb}{0.55,0.0,0.0}

\newcounter{addfl}

\newtheorem{theorem}{Theorem}
\newtheorem{proposition}[theorem]{Proposition}
\newtheorem{corollary}[theorem]{Corollary}
\newtheorem{lemma}[theorem]{Lemma}

\theoremstyle{definition}
\newtheorem{definition}[theorem]{Definition}
\newtheorem{example}[theorem]{Example}

\begin{document}

\renewcommand{\textfraction}{.002}
\renewcommand{\floatpagefraction}{.995}
\setcounter{dbltopnumber}{2}
\setcounter{totalnumber}{5}
\renewcommand{\topfraction}{.995}
\renewcommand{\dbltopfraction}{.995}
\renewcommand{\dblfloatpagefraction}{.995}

\title{What makes a reaction network ``chemical''?}

\author{Stefan M\"uller$^1$, Christoph Flamm$^2$, Peter F.~Stadler$^{3,4,5,2,6,7*}$}

\date{\footnotesize
$^1$~Faculty of Mathematics, University of Vienna, Oskar-Morgenstern-Platz 1, A-1090 Wien, Austria.
$^2$~Department of Theoretical Chemistry, University of Vienna, W\"ahringer Stra{\ss}e 17, A-1090 Wien, Austria.
$^3$~Bioinformatics Group, Department of Computer Science, and Interdisciplinary Center for Bioinformatics, 
Universit\"at Leipzig, H\"artelstra{\ss}e 16-18, D-04107, Leipzig, Germany.
$^4$~German Centre for Integrative Biodiversity Research (iDiv) Halle-Jena-Leipzig \& Competence Center for Scalable Data Services and
Solutions Dresden-Leipzig \& Leipzig Research Center for Civilization Diseases, University Leipzig, D-04107 Leipzig, Germany. 
$^5$~Max Planck Institute for Mathematics in the Sciences, Inselstra{\ss}e 22, D-04103 Leipzig, Germany.
$^6$~Faculdad de Ciencias, Universidad Nacional de Colombia, Sede Bogot\'a, Ciudad Universitaria, COL-111321 Bogot\'a, D.C., Colombia.
$^7$~Santa Fe Institute, 1399 Hyde Park Road, NM87501 Santa Fe, USA.
$^*$~Correspondence: studla@bioinf.uni-leipzig.de \\[3ex]
\normalsize \today}

%\maketitle

\twocolumn[
  \begin{@twocolumnfalse}
  \maketitle
\begin{abstract} 
  {\bf Background:}
  Reaction networks (RNs) comprise a set $X$ of species and a set $\RR$ of
  reactions $Y\to Y'$, each converting a multiset of educts $Y\subseteq X$
  into a multiset $Y'\subseteq X$ of products. RNs are equivalent to
  directed hypergraphs.  However, not all RNs necessarily admit a chemical
  interpretation. Instead, they might contradict fundamental principles of
  physics such as the conservation of energy and mass or the reversibility
  of chemical reactions. The consequences of these necessary conditions for
  the stoichiometric matrix $\SM \in \R^{X\times\RR}$ have been discussed
  extensively in the chemical literature. Here, we provide sufficient
  conditions for $\SM$ that guarantee the interpretation of RNs in terms of
  balanced sum formulas and structural formulas, respectively.

  {\bf Results:}
  Chemically plausible RNs allow neither a perpetuum mobile, i.e., a
  ``futile cycle'' of reactions with non-vanishing energy production, nor
  the creation or annihilation of mass. Such RNs are said to be
  thermodynamically sound and conservative. For finite RNs, both conditions
  can be expressed equivalently as properties of the stoichiometric matrix
  $\SM$. The first condition is vacuous for reversible networks, but it
  excludes irreversible futile cycles and -- in a stricter sense -- futile
  cycles that even contain an irreversible reaction. The second condition
  is equivalent to the existence of a strictly positive reaction
  invariant. Furthermore, it is sufficient for the existence of a
  realization in terms of sum formulas, obeying conservation of
  ``atoms''. In particular, these realizations can be chosen such that any
  two species have distinct sum formulas, unless $\SM$ implies that they
  are ``obligatory isomers''.  In terms of structural formulas, every
  compound is a labeled multigraph, in essence a Lewis formula, and
  reactions comprise only a rearrangement of bonds such that the total bond
  order is preserved. In particular, for every conservative RN, there
  exists a Lewis realization, in which any two compounds are realized by
  pairwisely distinct multigraphs.  Finally, we show that, in general,
  there are infinitely many realizations for a given conservative RN.
  
  {\bf Conclusions:}
  ``Chemical'' RNs are directed hypergraphs with a stoichiometric matrix
  $\SM$ whose left kernel contains a strictly positive vector and whose
  right kernel does not contain a futile cycle involving an irreversible
  reaction. This simple characterization also provides a concise
  specification of random models for chemical RNs that additionally
  constrain $\SM$ by rank, sparsity, or distribution of the non-zero
  entries. Furthermore, it suggests several interesing avenues for future
  research, in particular, concerning alternative representations of
  reaction networks and infinite chemical universes.
  
  \vspace{2ex}
\end{abstract}
\vspace{1cm}
\end{@twocolumnfalse}
]

%
%\begin{keyword}
%  \kwd{chemical reaction network}
%  \kwd{directed hypergraph}
%  \kwd{stoichiometrix matrix}
%  \kwd{futile cycle}
%  \kwd{perpetuum mobile}
%  \kwd{energy conservation}
%  \kwd{mass conservation}
%  \kwd{reaction invariants}
%  \kwd{null spaces}
%  \kwd{sum formula}
%  \kwd{multigraph}
%  \kwd{Lewis formula}
%\end{keyword}
%
%\end{abstractbox}
%\end{fmbox}% uncomment this for twcolumn layout
%\end{frontmatter}

%%%%%%%%%%%%%%%%%%%%%%%%% start of article main body

\section{Background}
%\smallskip

Most authors will agree that a chemical reaction network RN consists of a
set $X$ of chemical species or compounds and a set $\RR$ of chemical
reactions, each describing the transformation of some (multi)set of
  educts into a (multi)set of products. Depending on the application, this
basic construction may be augmented by assigning properties such as mass,
energy, sum formulas, or structural formulas to the compounds.  Similarly,
reactions may be associated with rate constants, equilibrium constants, and
so on. A formal theory of reaction networks (RN) describes a reaction
on a given set of compounds $X$ as a \emph{stoichiometric relation}, i.e.,
as a pair of formal sums of chemical species $x \in X$:
\begin{equation}
  \sum_{x\in X} s^-_{xr} \, x \to  \sum_{x\in X} s^+_{xr} \, x .
\label{eq:reaction}
\end{equation}
The left-hand side in equ.(\ref{eq:reaction}) lists the educts and the
right-hand side gives the products of the reaction. The
\emph{stoichiometric coefficients} $s^-_{xr}\in\N_0$ and $s^+_{xr}\in\N_0$
denote the number of species $x\in X$ that are consumed (on the left-hand
side) or produced (on the right-hand side) by the reaction $r$,
respectively. A species $x\in X$ is an educt in reaction $r$ if
$s^-_{xr}>0$ and a product if $s^+_{xr}>0$. If $s^+_{xr}=s^-_{xr}=0$, then
species $x$ does not take part in reaction $r$ and is suppressed in the
conventional chemical notation. The formal sums
  $\sum_{x\in X} s^-_{xr} \, x$ and $\sum_{x\in X} s^+_{xr} \, x$ form the
  \emph{complexes} of educts $r^-$ and products $r^+$ of the reaction $r$.
We denote the set of reactions under considerations by $\RR$ and call the
  pair $(X,\RR)$ a reaction network (RN). Throughout this contribution we
  will assume that both $X$ and $\RR$ are non-empty and finite. Excluding
explicit catalysis, that is, forbidding $s^-_{xr} \, s^+_{xr}>0$, is
suffices to consider the \emph{stoichiometric matrix}
$\SM \in \N_0^{X \times \RR}$. Its entries $\SM_{xr} = s^+_{xr} - s^-_{xr}$
describe the net production or consumption of species $x$ in reaction
$r$. In many practical applications, e.g.\ in the context of metabolic
networks, RNs are embedded in an open system. In that manner, the
consumption of nutrients and the production of waste can be modeled. We will
return to this point only after discussing chemical RNs in isolation, i.e.,
as closed systems.

Several graph representations have been considered as (simplified) models
of a RN, see \cite{Sandefur:13} for a recent summary. In contrast to the
pair $(X,\RR)$, they do not always completely represent the RN.

The S-graph (\emph{species graph}, \emph{compound graph}, or
\emph{substrate network} in the context of metabolic networks) has the
species as its vertices. A (directed) edge connects $x$ to $y$ if the RN
contains a reaction that has $x$ as an educt and $y$ as a product
\cite{Alon:07,Shellman:13}.  The corresponding construction in the kinetic
setting is the \emph{interaction} graph with undirected edges whenever
$\partial [x]/\partial[y]\not\equiv 0$, which are usually annotated by the
sign of the derivative \cite{Soule:03}. S-graphs have also proved to be
useful in approximation algorithms for the minimal seed set problem
\cite{Borenstein:08}, which asks for the smallest set of substrates that
can generate all metabolites. Complementarily, \emph{reaction graphs} model
reactions as nodes, while edges denote shared molecules
\cite{Fagerberg:13a}.

The \emph{complex-reaction graph} simply has the complexes $\mathscr{C}$
(the left- and right-hand sides of the reactions) as its vertex set and the
reactions $\RR$ as its edge set.  That is, two complexes $r^-$ and $r^+$
are connected by a directed edge if there is a reaction
$r=(r^-,r^+)\in\RR$.  Its incidence matrix
$\mathbf{Z} \in \RR^{\mathscr{C}\times\RR}$ (with entries
$\mathbf{Z}_{cr}=-1$ if $c=r^-$, $\mathbf{Z}_{cr}=1$ if $c=r^+$, and
$\mathbf{Z}_{cr}=0$ otherwise) is linked to the stochiometric matrix via
$\mathbf{S}=\mathbf{Y}\mathbf{Z}$, where the entries of the
(stoichiometric) complex matrix $\mathbf{Y} \in \R^{X \times \mathscr{C}}$
are the corresponding stochiometric coefficients.  The complex-reaction
graph plays a key role in the analysis of chemical reaction networks with
mass-action kinetics and arbitrary rate constants, as studied in classical
``chemical reaction network theory''
(CRNT)~\cite{Horn:72,Horn:72a,Feinberg:72}.  It gives rise to notions such
as ``complex balancing'' and ``deficiency'', which allow the formulation of
strong (global) stability results, see e.g.\ \cite{Craciun:09,Angeli:09}.

SR-graphs (\emph{Species-reaction networks}) are bipartite graphs with
different types of nodes for chemical species and reactions, respectively
\cite{Craciun:06,Kaltenbach:12}. As such, they can be endowed with
additional annotations or extended with multiple edges to represent
stoichiometric coefficients. In this extended form, they are faithful
representations of chemical RNs. Alternatively, the edges are often
annotated with the multiplicities of molecules, i.e., the stoichiometric
coefficients; in this case, they completely specify the RN $(X,\RR)$.
Undirected SR-graphs have a close relationship to classical deficiency
theory \cite{Feinberg:72,Horn:72} and form the starting point for a
qualitative theory of chemical RN kinetics \cite{Shinar:13}). More detailed
information on qualitative kinetic behavior can be extracted from directed
SR-graphs \cite{Mincheva:06}. Both the S- and the R-graph can be extracted
unambigously from an SR-graph.
 
The bipartite SR-graphs can be interpreted as the K{\"o}nig's
representation \cite{Zykov:74} of directed hypergraphs. The connection
between hypergraph and graph representations is discussed in some more
detail in \cite{Zhou:11}. While SR-graphs and directed hypergraphs can be
transformed into each other, they carry a very different semantic. For
instance, the notions of path and connectivity are very different for
bipartite graphs and directed hypergraphs \cite{SantiagoArguello:21a}. It
has been argued, therefore that any graph representation of chemical
networks necessarily treats edges as independent entities and thus fails to
correctly capture the nature of chemical reactions
\cite{Klamt:09,Montanez:10}. In a similar vein, \cite{Andersen:19a} adopts
the hypergraph representation and models (bio)chemical pathways as integer
hyperflows to ensure mass balance at each vertex. Not every pair of an S-
and R-graph implies an SR-graph, and if they do, the result need not be
unique \cite{Fagerberg:13a}.

Over the last decade, many authors, including one of us, have investigated
metabolic networks from a statistical perspective and reached the
conclusion that they are distictly ``non-random'', presumably as the
consequence of four billion years of evolution.  This conclusion is
typically reached by first converting a RN into one of the graph
representations mentioned above. The choice of graphs is largely motivated
by a desire to place metabolic or other chemical RNs within the scheme of
small world and scale free networks and to analyze the RNs with the
well-established tools of network science \cite{Wagner:01,Klamt:09}. Thus
one concludes that graph-theoretical properties of metabolic networks are
significantly different from the properties of randomly generated or
randomized background models for chemical reaction networks
\cite{Jeong:00,Gleiss:01a,Shellman:13,Fischer:15}.  The insights gained
from this ``non-randomness'' of metabolism, however, critically depend on
what exactly the authors meant by ``random'', that is, how the background
models are defined. In particular, it is important to understand whether
differences between chemical networks and the background are caused by the
implementation of universal properties (that any ``\chemlike'' RN must
satisfy) or whether they arise from the intrinsic structure of particular
chemical networks.

To this end, however, we first need a comprehensive conception of what
constitutes a \emph{\chemlike} reaction network. The different
representations used in the literature highlight the fact that it is far
from obvious which graphs or hypergraphs properly describe \emph{chemical}
RNs among a possibly much larger set of network models.  There is a
significant body of work in the literature that describes necessary
conditions on the stoichiometric matrix $\SM$ that derive from key
properties of chemical RNs, such as the conservation of mass or atoms in
each reaction \cite{Horn:72a,Schuster:91,Gadewar:01,%
  Famili:03,Flockerzi:07,Haraldsdottir:18}. In contrast, we are interested
here in sufficient conditions with the aim of providing a concise
characterization of RNs $(X,\RR)$ and their stoichiometric matrices $\SM$
that describe reaction system that can reasonably be considered as
``\chemlike''. This is of practical relevance in particular for the
construction of artifical chemistry models
\cite{Fontana:91a,Dittrich:01,Benkoe:03b,Banzhaf:15} and random
``chemistries''. In particular, it is still an open problem how random RNs
can be constructed that can serve as fair, \chemlike background models. We
therefore start with a brief survey of random artificial chemistries and
randomized RNs. As we shall see, no explicit provisions are made to include
``chemical'' constraints such as the conservation of matter and energy into
the background models.

The main part of this contribution is the characterization of \chemlike
RNs. Starting from the principles of energy conservation and conservation
of matter, we derive equivalent conditions on the stoichiometric matrix
$\SM$. We then introduce realizability of RNs by sum formulas and
structural formulas as a first step towards a formalization of \chemlike
networks, and show that conservation of matter is already sufficient to
guarantee the existence of such \chemlike representations. Finally we
discuss the consequences of the mathematical results for the construction
of random RNs.
  
\section{A brief survey of random and randomized chemical RNs}
%\smallskip

Chemical reaction networks are specified either as a set of chemical
reactions or as a system of differential equations describing its
kinetics. Graphical models have been extracted from both. 

\subsection{Simple graph models of RNs}
%\smallskip

S-graphs have been used to explore statistical properties of large RNs.
%% redundent sentence! same as page 2 end of last paragraph This is largely
%% motivated by the desire to place chemical reaction networks within the
%% scheme of small world and scale free
%% networks. \cite{Wagner:01,Klamt:09}.
In this line of research, empirical S-graphs are
compared to the ``usual'' random networks models such as Erd{\H{o}}s
Reny{\'\i} (ER) random graphs, Small World networks in the sense of Watts
and Strogatz \cite{Watts:98}, or the {\'A}lbert-Barabasi model of
preferential attachment. Generative models for random graphs with given
degree distributions were introduced in \cite{Newman:01}.  Not
surprisingly, chemical reaction networks do not very well conform to either
one of them. As noted early on, however, R-graphs of metabolic networks at
least qualitatively fit the small world paradigm \cite{Wagner:01}. More
sophisticated analyses detected evidence for modularity and hierarchical
organization in metabolic networks \cite{Ravasz:02}, using random graph
models with the same degree distributions as contrasts.  Arita noted,
however, that S-graphs are poor representations of biochemical pathways and
proposed an analysis in terms of atom traces, concluding that ``the metabolic
world [of \emph{E.\ coli}] is not small in terms of biosynthesis and
degradation'' \cite{Arita:04}. The motivation to focus on atom maps
  comes from the insight that two compounds that are linked by reactions
  are only related by the chemical transformation if they share at least
  one atom.

A versatile generator for bipartite graphs that can handle joint degree
distributions is described in \cite{Azizi:17}. Surprisingly, bipartite
random graph models apparently have not been used to model chemistry.
Instead of generative models such ER graph or the preferential attachment
model, null models are often specified in terms of rewiring, that is, edit
operations on the graph. Rewiring rules define a Markov Process on a set of
graphs that can produce samples of randomized networks. The key idea is to
specify the rewiring procedure in such a way that it preserves graph
properties that are perceived to be important \cite{Rao:96,Hanhijaervi:09}.
For example, the degrees of all vertices in a digraph are preserved when a
pair of directed edges $x_1y_1$ and $x_2y_2$ is replaced by $x_1y_2$ and
$x_2y_1$ as long as $x_1$ and $x_2$ have the same out-degree while $y_1$
and $y_2$ have the same in-degree. Randomization procedures for bipartite
graphs have become available in the context of ecological networks
\cite{Strona:14} or trade networks \cite{Saracco:15}. To our knowledge they
have not been used for SR graphs.

\subsection{Random (Directed) Hypergraphs}
%\smallskip

In \cite{dePanafieu:15} a hypergraph is defined as a multiset of
hyperedges, each of which in turn is a multiset of vertices. In this
setting, a random hypergraph is specified by the probabilities $p_k$ to
include a hyperedge $e$ with cardinality $|e|=k$. Similar models for
undirected hypergraphs are used e.g.\ in \cite{Ghoshal:09}.  In a directed
hypergraph, every hyperedge is defined as the pair $(e_-,e_+)$ consisting
of the multisets $e_-$ and $e_+$. The construction of \cite{dePanafieu:15}
thus naturally generalizes to directed hypergraphs specified by picking $e$
with probability $p_{|e_-|,|e_+|}$. In the context of chemistry this amounts
to picking educt and product sets for reactions with probabilities
depending on their cardinality. This type of random (directed) hypergraph
models are the obvious generalizations of the Erd{\H{o}}s Reny{\'\i}
(di)graphs. A certain class of random directed hypergraphs with $|e^-|=2$
and $|e^+|=1$ for all hyperedges $e$ is considered in \cite{Sloan:12}.

Hypergraphs are also amenable to rewiring procedures that ensures the
preservation of certain local or global properties. For instance,
\cite{Zhou:11} proposes a scheme that preserves the number and cardinality
of the hyperedges (replacing a randomly selected $(e_-,e_+)$ with a
randomly selected pair of disjoint subsets $(e_-',e_+')$ with
$|e_-|=|e_-'|$ and $|e_+|=|e_+'|$). On this basis, the authors conclude
that the hierarchical structure hypothesis proposed in \cite{Ravasz:02} is
not supported for metabolic networks when a clustering coefficient is
defined for directed hypergraphs.  \cite{Zhou:11} also compares S- and
R-graphs of metabolic networks with ensembles of S- and R-graphs derived
from randomized directed hypergraphs and cast further doubt on previously
reported scaling results. Randomization procedures for hypergraphs that
preserve local clustering are described in \cite{Nakajima:21}. An approach
that uses a chemical graph rewriting model to ensure soundness of reactions
is described in the MSc thesis \cite{Braun:19}.

In \cite{Fischer:15} networks are constructed in a stepwise procedure
starting with directed graphs whose arcs are then re-interpreted as
directed hyperarcs by combining multiple arcs. This process is guided by
matching the degree distribution of the implied S-graph.

\subsection{Reaction Universes: Random Subhypergraphs} 
%\smallskip

Instead of generating a random RN directly from a statistical model or
rewiring a given one, one can also start from a \emph{reaction universe}
RU, that is, a RN that contains all species of interest and all known or
inferred reactions between them. Without losing generality we can think of
the RU as a directed hypergraph in the sense of \cite{dePanafieu:15}, where
the multi-set formalism accounts for the stoichiometric coefficients.  In
contrast to the generative and rewiring approaches the \emph{a priori}
specification of an RU ensures a high level of chemical realism and RNs can
now be sampled by randomly selecting subsets of directed hyperedges, that
is, chemical reactions. If the RU already ensures conservation of
  matter or energy, these properties are inherited by the sub-networks.
In order to generate random metabolic networks, reactions can be drawn from
databases such as KEGG or EcoCyc \cite{Samal:10,Kim:19}. Such selections of
reactions are sometimes called ``metabolic genotypes'' since the available
reactions are associated with enzymes, whose presence or absence is
determined by an organism's genome \cite{Samal:10}. In some studies,
additional constraints such as the production of biomass are exploited and
networks are sampled e.g.\ by combining Flux-Balance Analysis (FBA) and a
Markov Chain Monte Carlo (MCMC) approach
\cite{MatiasRodrigues:09,Samal:10}.

\section{A characterization of \chemlike reaction networks} 
%\smallskip

In this section, we start from reaction networks that are specified as
abstract stoichiometric relations, Equ.~(\ref{eq:reaction}), and identify
minimal constraints necessary to avoid blatantly unphysical behavior. 

\subsection{Notation and Preliminaries}
%\smallskip

Let $X$ be a finite set and let $\RR$ be a pair of formal sums of elements
of $X$ with non-negative integer coefficients according to
Equ.~(\ref{eq:reaction}). Then we call the pair $(X,\RR)$ a \emph{reaction
  network} (RN). Equivalently, a RN is a directed, integer-weighted
hypergraph with directed edges $(r^-,r^+)$ such that $x\in r^-$ with weight
$s^-_{xr}>0$ and $x\in r^+$ with weight $s^+_{xr}>0$. The weights
$s^-_{xr}$ and $s^+_{xr}$ are usually called the \emph{stoichiometric
  coefficients}. We set $s^-_{xr}=0$ and $s^+_{xr}=0$ if $x\notin r^-$ and
$x\notin r^-$, respectively. We deliberately dropped the qualifier
\emph{chemical} here since, as we shall see, not every RN $(X,\RR)$ makes
sense as a model of a chemical system. In fact, the aim of this
contribution is to characterize the set of RNs that make sense as models of
chemistry.

\begin{figure*}
  \begin{center}
%    \begin{minipage}{0.6\textwidth}
      \includegraphics[width=0.6\textwidth]{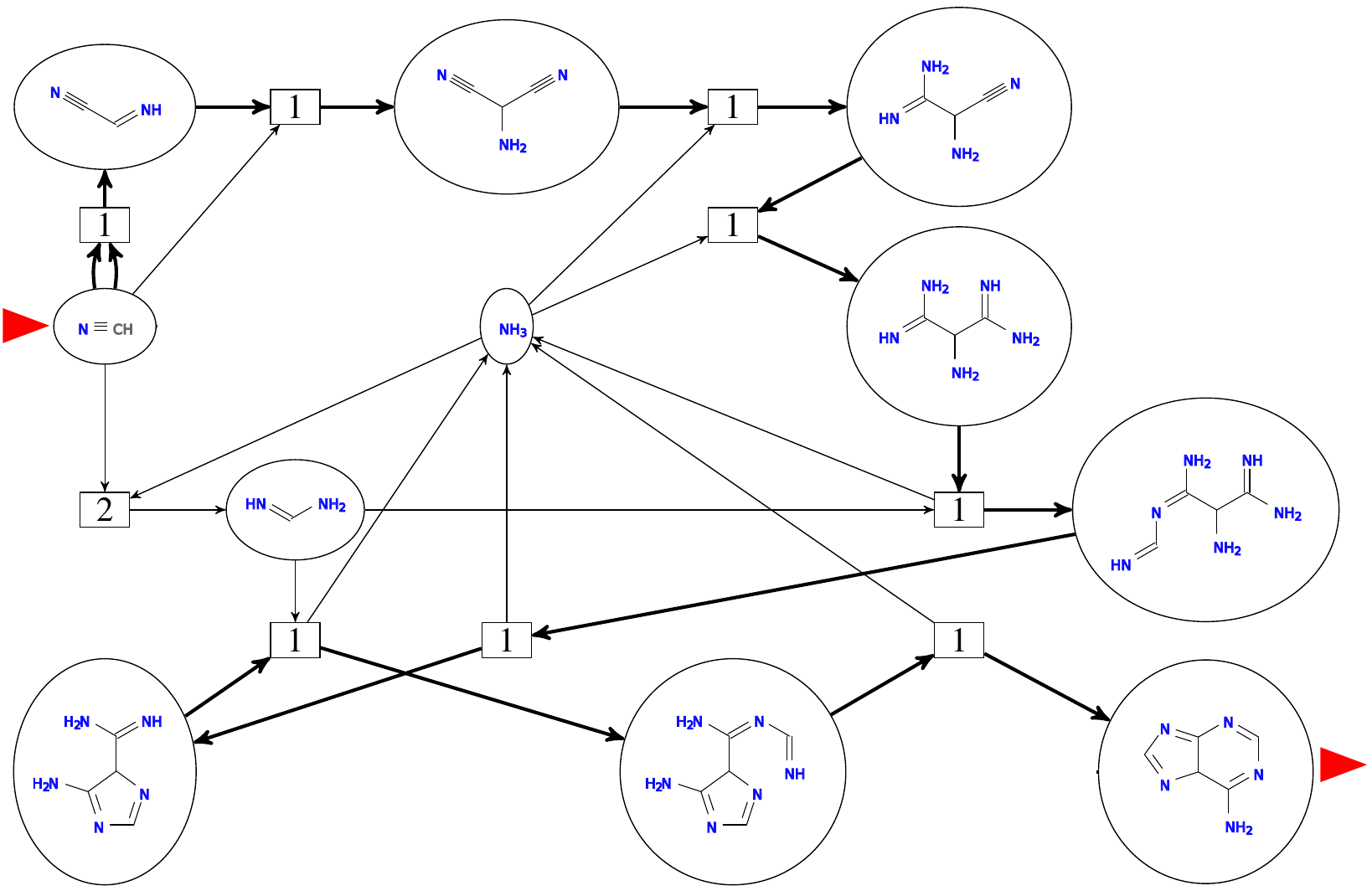}
%    \end{minipage}
%    \quad
%    \begin{minipage}{0.3\textwidth}
      \caption{Representation of a RN as K{\"o}nig multigraph of the
        corresponding directed hypergraph. Round vertices (with chemical
        structures shown inside) designate compounds $x\in X$, while
        reactions $r\in R$ are shown as square vertices. Stoichometric
        coefficients are indicated by the number of edges from $x$ to $r$
        for $s^-_{xr}>0$ and $r$ to $x$ for $s^+_{xr}>0$, respectively. A
        flow (an overall reaction) is given by non-negative integer
        multiples of individual reactions. Here the coefficients $v_r$ are
        indicated in the square nodes for each reaction $r$. The flow shown
        here defines Or{\'o}'s \cite{Oro61} route from \ce{HCN} to adenine
        (marked by red triangles) and corresponds to the net reaction \ce{5
          HCN \longrightarrow H5C5N5}. Figure adapted from
        \cite{Andersen:13b}.}
      \label{fig:hypergraph}
%    \end{minipage}
  \end{center}
\end{figure*}

Such directed hypergraphs are most conveniently drawn as (bipartite)
K{\"o}nig multigraphs, with distinct types of vertices representing
compounds $x\in X$ and reactions $r\in\RR$, respectively. Stoichiometric
coefficients larger than one appear as multiple edges. See the example in
Fig.~\ref{fig:hypergraph}.

For each reaction $r\in\RR$, we define its support as
$\supp(r)=\{x \mid s^-_{xr}+s^+_{xr}>0\}$; that is, $x\in \supp(r)$ if it
appears an educt, a product, or a catalyst in $r$.  The stoichiometric
matrix of $(X,\RR)$ is $\SM \in \N_0^{X \times \RR}$ with entries
$\SM_{xr}= s^+_{xr} - s^-_{xr}$.

We distinguish \emph{proper reactions} $r$, for which there is both
$x\in X$ with $s_{xr}<0$ and $y\in X$ with $s_{yr}>0$, \emph{import
  reactions} for which $s_{xr}\ge 0$ for all $x\in X$, and \emph{export
  reactions} for which $s_{xr}\le 0$ for all $x\in X$. We write
$\varnothing$ for the empty formula, hence \ce{$\varnothing$
  \longrightarrow A} and \ce{B \longrightarrow $\varnothing$} designate the
import of \ce{A} and the export of \ce{B}, respectively. Note that this
definition also allows catalyzed import and export reactions, e.g., \ce{C
  \longrightarrow C + A} or \ce{B + C \longrightarrow C}.

In thermodynamics, a system is closed if it does not exchange matter with
its environment. For a RN, this rules out import and export reactions.
\begin{definition}
  A RN $(X,\RR)$ is \emph{closed} if all reactions $r\in\RR$ are proper.
\end{definition}
Given an arbitrary RN $(X,\RR)$, there is a unique inclusion-maximal closed
RN contained in $(X,\RR)$, namely $(X,\RR^\textrm{p})$ with
\begin{equation} 
  \RR^p= \{r\in\RR \mid r \textrm{ is proper}\}.
\end{equation}
We will refer to $(X,\RR^p)$ as the \emph{proper part} of $(X,\RR)$.

For every reaction $r$, one can define a \emph{reverse reaction} $\bar r$
that is obtained from $r$ by exchanging the role of products and
educts. That is, $\bar r$ is the reverse of $r$ iff, for all $x\in X$, it
holds that
\begin{equation}
  s^-_{x\bar r} = s^+_{xr} \quad\textrm{and}\quad
  s^+_{x\bar r} = s^-_{xr} . 
\end{equation}
While thermodynamics dictates that every reaction is reversible in
principle (albeit possibly with an extremely low reaction rate), it is a
matter of modeling whether sufficiently slow reactions are included in the
reaction set $\RR$.

Chemical reactions can be composed and aggregated into ``overall
reactions''. In the literature on metabolic networks, \emph{pathways} are
of this form. An overall reaction consists of multiple reactions that
collectively convert a set of educts into a set of products. It can be
represented as a formal sum of reactions $\sum_{r\in \RR} \vv_r \, r$,
where the vector of multiplicities $\vv \in \N^\RR_0$ has non-negative
integer entries. Thereby, $[\SM\vv]_x$ determines the net consumption or
production of compound $x$ in the overall reaction specified by $\vv$.

A vector $\vv\in\N_0^\RR$ can be interpreted as an \emph{integer hyperflow}
in the following sense: If $x$ is neither an educt nor a product of the
overall reaction specified by~$\vv$, then
$[\SM\vv]_x = \sum_r (s^+_{xr}-s^-_{xr}) \vv_r = 0$, i.e., every unit of
$x$ that is produced by some reaction $r$ with $\vv_r>0$ is consumed by
another reaction $r'$ with $\vv_{r'}>0$.

The effect of an overall reaction can be represented via formal sums of
species in two ways: as \emph{composite reactions},
\begin{equation}
  \label{eq:compositeR}
  \sum_{x\in X} \left(\sum_{r\in\RR} s^-_{xr}\vv_r\right)  x \longrightarrow
  \sum_{x\in X} \left(\sum_{r\in\RR} s^+_{xr}\vv_r\right) x , 
\end{equation}
or as \emph{net reactions},
\begin{equation}
  \label{eq:netR}
  \begin{split}
    \sum_{x\in X}& \left(\sum_{r\in\RR}
      \left[(s^-_{xr}-s^+_{xr})\vv_r\right]_{+}\right) x \\&\longrightarrow
    \sum_{x\in X} \left(\sum_{r\in\RR}
      \left[(s^+_{xr}-s^-_{xr})\vv_r\right]_{+}\right) x .
  \end{split}
\end{equation}
Here we use the notation $[c]_+ = c$ if $c>0$ and $[c]_+=0$ for $c\le 0$.
In Equ.(\ref{eq:netR}), intermediates, i.e., formal catalysts are
cancelled.  Hence, the net consumption (or production) of a species $x$ is
$\sum_{r\in\RR}[(s^-_{xr}-s^+_{xr})\vv_r]_{+}=-[\SM\vv]_x$ if
$[\SM\vv]_x<0$ (or
$\sum_{r\in\RR}[(s^+_{xr}-s^-_{xr})\vv_r]_{+}=[\SM\vv]_x$ if
$[\SM\vv]_x>0$).
  
Fig.~\ref{fig:hypergraph} shows the RN of Oro's prebiotic adenine synthesis
from \ce{HCN} and the integer hyperflow $\vv$ corresponding to the net
reaction ``\ce{5 HCN $\longrightarrow$} adenine'' as an example.

While a restriction to integer hyperflows $\vv\in\N_0^\RR$ is necessary in
many applications, see e.g.\ \cite{Andersen:19a} for a detailed discussion,
it appears mathematically more convenient to use the more general setting
of \emph{fluxes} $\vv\in\R^\RR_\ge$ as in the analysis of metabolic
pathways. To emphasize the connection with the body of literature on
network (hyper)flows we will uniformly speak of flows.

For any vector $\vv \in \R^\RR$, we write $\vv\ge 0$ if $\vv$ is
non-negative, $\vv>0$ if $\vv$ is non-negative and non-zero, that is, at
least one entry is positive, and $\vv\gg0$ if all entries of $\vv$ are
positive. Analogously, we write $\vv\le 0$, $\vv<0$, and $\vv\ll 0$. In
particular, a vector $\vv\in\R^\RR$ is called a \emph{flow} if $\vv\ge 0$.

A \emph{non-trivial} flow satisfies $\vv>0$, i.e., $\vv\ne0$.  Two flows
$\mathbf{v_1}$ and $\mathbf{v_2}$ are called \emph{parallel} if they
describe the same net reaction.  In particular, we therefore have
$\SM\mathbf{v_1} = \SM\mathbf{v_2}$ for parallel flows.

Futile cycles in a RN are non-trivial flows for which educts and products
coincide and thus the net reaction is empty.
\begin{definition}
  A flow $\vv>0$ is a \emph{futile cycle} if $\SM\vv=0$.
\end{definition}
We use the term futile cycle in the strict sense to describe the concurrent
activity of multiple reactions (or pathways) having no net effect other
than the dissipation energy.  In the literature on metabolic networks often
a less restrictive concept is used that allows certain compounds (usually
co-factors, ATP/ADP, redox equivalents, or solvents) to differ between
products and educts, see e.g.\
\cite{Schilling2000,Beard2002,Schwender:04,Qian:06}.  In this setting, the
net reaction of concurrent glycolysis and gluconeogenesis, namely the
hydrolysis of ATP, is viewed as energy dissipation rather than a chemical
reaction. In our setting, \ce{ATP + H2O -> ADP + P$_i$- + H+}, is a net
reaction like any other, and hence a futile cycle would only arise if
recycling of ATP, i.e., \ce{ADP + P$_i$- + H+ -> ATP + H2O}, was included
as well.

If a RN has a futile cycle, it also has an integer futile cycle
$\vv\in\N_0^\RR$, since $\SM$ has integer entries and thus its kernel has a
rational basis, which can be scaled with the least common denominator to
have integer entries.

A pair $(X',\RR')$ is a \emph{subnetwork} of $(X,\RR)$ if $X'\subseteq X$,
$\RR'\subseteq\RR$, and $\supp(r)\subseteq X'$ implies $r\in\RR'$.  We say
that a property $P$ of a RN is \emph{hereditary} if ``$(X,\RR)$ has $P$''
implies that every subnetwork ``$(X',\RR')$ has $P$''.

Chemical reactions are subject to thermodynamic constraints that are a
direct consequence of the conservation of energy, the conservation of mass,
and the reversibility of chemical reactions. In the context of chemistry,
conservation of mass is of course a consequence of the conservation of
atoms throughout a chemical reaction. In the following sections, we
investigate how these physical principles constrain RNs.  Since we have
introduced RNs in terms of abstract molecules and reactions,
Equ.~(\ref{eq:reaction}), we express the necessary conditions in terms of
the stoichiometric matrix~$\SM$, which fully captures only the proper part
of the RN.  Throughout this work, therefore, \emph{we assume that $(X,\RR)$
  is a closed RN, unless explicitly stated otherwise.}

\subsection{Thermodynamic constraints} 
%\smallskip

\subsubsection*{Reaction energies and perpetuum mobiles}
%\smallskip

Every chemical reaction $r$ is associated with a change in the Gibbs free
energy of educts and products. We therefore introduce a vector of
\emph{reaction (Gibbs free) energies} $\gf \in \R^\RR$ and write
$(X,\RR,\gf)$ for a RN endowed with reaction energies. The reaction energy
for an overall reaction is the total energy of the individual reactions
involved. In terms of $\vv\in\R^\RR$, it can be expressed as
\begin{equation}
  \sum_{r\in\RR} \gf_r\vv_r = \gf^\trans \vv =
  \langle \gf,\vv\rangle ,
\end{equation}
where $\langle\cdot,\cdot\rangle$ denotes the scalar product on $\R^\RR$.

Futile cycles \emph{may} act as a chemical version of a perpetuum
mobile. This is the case whenever a flow $\vv > 0$ with zero formal net
reaction, $\SM \vv = 0$, increases or decreases energy, i.e., if
$\langle \gf,\vv\rangle \neq 0$.
\begin{definition} %\label{def:pm} 
  Let $(X,\RR,\gf)$ be a RN with reaction
  energies.  A flow $\vv > 0$ is a \emph{perpetuum mobile} if $\SM\vv=0$
  and $\langle\gf,\vv\rangle \neq 0$.
\end{definition}
The classical concept of a perpetuum mobile decreases its energy,
$\langle\gf,\vv\rangle < 0$, thereby ``creating'' energy for its
environment.  An ``anti'' perpetuum mobile with $\langle\gf,\vv\rangle > 0$
would ``annihilate'' energy. Either situation violates energy conservation
and thus cannot be allowed in a chemical RN. Obviously, there is no
perpetuum mobile if $(X,\RR)$ does not admit a futile cycle.

In fact, thermodynamics dictates that Gibbs free energy is a state
function. Two parallel flows $\mathbf{v_1}$ and $\mathbf{v_2}$ therefore
must have the same associated net reaction energies. That is,
$\SM\mathbf{v_1}=\SM\mathbf{v_2}$ implies
$\langle \gf, \mathbf{v^1}\rangle = \langle
\gf,\mathbf{v^2}\rangle$. Equivalently, any vector
$\vv=\mathbf{v^1}-\mathbf{v^2} \in \R^\RR$ with $\SM\vv=0$ must satisfy
$\langle \gf,\vv\rangle =0$.  That is, $\gf \in (\ker \SM)^\perp$.
\begin{definition}
  \label{def:thdyn}
  Let $(X,\RR,\gf)$ be a RN with reaction energies.  Then $(X,\RR,\gf)$ is
  \emph{thermodynamic} if $\vv\in\R^\RR$ and $\SM\vv=0$ imply
  $\langle \gf,\vv\rangle =0$, that is, if $\gf \in (\ker \SM)^\perp$.
\end{definition}
Let $(X,\RR,\gf)$ be thermodynamic, $(X',\RR')$ be a subnetwork of
$(X,\RR)$, and $\gf'$ be the restriction of $\gf$ to $\RR'$. Then
$\vv'\in\R^{\RR'}$ corresponds to $\vv\in\R^{\RR}$ with
$\supp(\vv)\subseteq\RR'$, and thus $\vv'\in\R^{\RR'}$ and $\SM'\vv'=0$
imply $\SM\vv=0$ and further
$\langle \gf',\vv'\rangle=\langle \gf,\vv\rangle =0$. Hence
$(X',\RR',\gf')$ is again thermodynamic.

We note that the reaction energies of a reaction $r$ and its reverse
$\bar r$ necessarily cancel:
\begin{lemma} \label{fact:g-reverse}
  If $r$ and $\bar r$ are reverse reactions in a thermodynamic network
  $(X,\RR,\gf)$, then $\gf_{\bar r}=-\gf_r$.
\end{lemma}
\begin{proof}
  If $r$ and $\bar r$ are reverse reactions, then $\vv$
  with $\vv_r=\vv_{\bar r}=1$ (and $\vv_{r'}=0$
  otherwise) satisfies $\SM\vv=0$. Thus
  $\langle\gf,\vv\rangle = \gf_r+\gf_{\bar r}=0$.
\end{proof}

\subsubsection*{Digression: molecular energies and {Hess' Law}}
%\smallskip

Every molecular species $x\in X$ has an associated \emph{Gibbs free energy
  of formation}. For notational simplicity, we write $\GF_x$ instead of the
commonly used symbol $G_\mathrm{f}(x)$. The corresponding vector of
\emph{molecular energies} is denoted by $\GF \in \R^X$.  Molecular energies
and reactions energies $\gf \in \R^\RR$ are related by Hess' law: For every
reaction $r \in \RR$, it holds that
\begin{equation*}
  \gf_r = \sum_{x\in X} \GF_x (s^+_{xr}-s^-_{xr}) =
  \sum_{x\in X} \GF_x \, \SM_{xr} .
\end{equation*}
In matrix form, the relationship between reaction energies $\gf$ and
molecular energies $\GF$ amounts to 
\begin{equation} \label{eq:Hess}
\gf = \SM^\trans \GF .
\end{equation}

\begin{proposition}
  \label{pro:thermodynamic}
  Let $(X,\RR)$ be a RN and $\gf\in\R^\RR$ be a vector of reaction
  energies. Then $(X,\RR,\gf)$ is thermodynamic if and only if there
  is a vector of molecular energies $\GF\in\R^X$ satisfying Hess'
  law, Equ.~\eqref{eq:Hess}.
\end{proposition}
\begin{proof}
  By Definition~\ref{def:thdyn}, $(X,\RR,\gf)$ is thermodynamic if
  $\gf \in (\ker \SM)^\perp = \im \SM^\trans$, that is, if there is $\GF$
  such that $\gf = \SM^\trans \GF$, satisfying Hess's law.
\end{proof}
Note that the vector of molecular energies $\GF$ is \emph{not}
uniquely determined by $\gf$ in general. 

\subsubsection*{Reversible and irreversible networks}
%\smallskip

To begin with, we consider purely reversible or irreversible RNs.
\begin{definition} 
  A RN $(X,\RR)$ is \emph{reversible} if $r\in\RR$ implies $\bar r\in\RR$
  and \emph{irreversible} if $r\in\RR$ implies $\bar r\notin\RR$.
\end{definition}
In reversible networks, general vectors $\vv \in \R^\RR$ have corresponding
flows $\mathbf{\tilde v} \ge 0$ with the same net reactions and, in the
case of thermodynamic networks, with the same energies.
\begin{lemma}
  \label{lem}
  Let $(X,\RR,\gf)$ be a reversible RN (with reaction energies), and let
  $\vv \in \R^\RR$ be a vector.  Then there is a flow
  $\mathbf{\tilde v}\ge0$ such that $\SM \mathbf{\tilde v} = \SM \vv$.  If
  $(X,\RR,\gf)$ is thermodynamic, then further
  $\langle \gf , \mathbf{\tilde v} \rangle = \langle \gf , \vv \rangle$.
\end{lemma}
\begin{proof}
  If $\vv \ge 0$, there is nothing to show. Otherwise, there are two flows
  $\mathbf{v^1} \ge 0$ and $\mathbf{v^2} > 0$ such that
  $\vv=\mathbf{v^1}-\mathbf{v^2}$. Since $(X,\RR)$ is reversible, each
  reaction $r\in\RR$ has a reverse $\bar r$, and we define the reverse flow
  $\mathbf{\bar v^2}>0$ such that
  $\mathbf{\bar v^2}_{r} = \mathbf{v}^\mathbf{2}_{\bar r}$. By
  construction, it satisfies $\SM \mathbf{\bar v^2} = - \SM \mathbf{v^2}$.
  %and thus $\SM (\mathbf{v^2}+\mathbf{\bar v^2}) = 0$. 
  Now consider the flow
  $\mathbf{\tilde v} = \mathbf{v^1}+\mathbf{\bar v^2} > 0$.
  It satisfies
  $$\SM\mathbf{\tilde v}= \SM (\mathbf{v^1}+\mathbf{\bar v^2}) = \SM
  (\mathbf{v^1} - \mathbf{v^2}) = \SM \vv .$$ If the network is
  thermodynamic, then the reverse flow satisfies
  $\langle \gf , \mathbf{\bar v^2} \rangle = - \langle \gf, \mathbf{v^2}
  \rangle$, by Lemma~\ref{fact:g-reverse}.  Hence,\\
  $\displaystyle \langle \gf , \mathbf{\tilde{v}} \rangle = \langle \gf ,
  \mathbf{v^1} + \mathbf{\bar v^2} \rangle = \langle \gf , \mathbf{v^1} -
  \mathbf{v^2} \rangle = \langle \gf , \vv \rangle$.
\end{proof}

By definition, a thermodynamic network cannot contain a perpetuum mobile.
Conversely, by the result below, if a reversible network is not
thermodynamic, then it contains a perpetuum mobile.

\begin{proposition}
  \label{prop:TD=PM}
  Let $(X,\RR,\gf)$ be a reversible RN with reaction energies.
  Then, the following two statements are equivalent:
  \begin{itemize}
  \item[(i)] $(X,\RR,\gf)$ is thermodynamic.
  \item[(ii)] $(X,\RR,\gf)$ contains no perpetuum mobile.
  \end{itemize}
\end{proposition}
\begin{proof}
  Suppose $(X,\RR,\gf)$ is not thermodynamic.  That is, there is
  $\vv \in \R^\RR$ with $\SM\vv=0$ and $\langle \vv, \gf \rangle \neq 0$.
  By Lemma~\ref{lem}, there is $\mathbf{\tilde v}\ge0$ with
  $\SM \mathbf{\tilde v}=0$ and
  $\langle \mathbf{\tilde v}, \gf \rangle \neq 0$, that is, a perpetuum
  mobile.
\end{proof}

The exclusion of a perpetuum mobile is not sufficient in non-reversible
systems.
\begin{example} 
Consider the following RN (with reaction
energies $\gf$):
\begin{equation} 
  \begin{split}
    1\colon \quad \ce{A}&\longrightarrow \ce{B} , \quad \gf_1=-1 , \\
    \bar{1}\colon \quad \ce{B}&\longrightarrow \ce{A} ,
    \quad \gf_{\bar 1}=+1 , \\
    2\colon \quad \ce{B}&\longrightarrow \ce{C} , \quad \gf_2=-1 , \\
    3\colon \quad \ce{A}&\longrightarrow \ce{C} , \quad \gf_3=-1 .
  \end{split}
  \label{eq:example1}
\end{equation}
It contains one futile cycle, \\
$\ce{A} \stackrel 1 \to \ce{B} \stackrel {\bar 1} \to \ce{A}$,
$\vv=(1,1,0,0)^\trans$ with $\langle \gf,\vv\rangle = 0$, \\
but no perpetuum mobile.
However, it contains two parallel flows with different energies, \\
$\ce{A} \stackrel 1 \to \ce{B} \stackrel 2 \to \ce{C}$,
$\vv=(1,0,1,0)^\trans$ with $\langle \gf,\vv\rangle = -2$, \\
$\ce{A} \stackrel 3 \to \ce{C}$, $\vv=(0,0,0,1)^\trans$ with
$\langle \gf,\vv\rangle = -1$. \\
Hence, the RN (with reaction energies $\gf$) is not thermodynamic.  By
setting $\gf_3=-2$, it can be made thermodynamic.
\end{example}

Many RN models are non-reversible, i.e., they contain
irreversible reactions whose reverse reactions are so slow that they
are neglected. From a thermodynamic perspective, irreversible reactions $r$
must be \emph{exergonic}, i.e., $\gf_r<0$. We first consider the extreme
case that all reactions $r\in\RR$ are irreversible.

\begin{proposition} %\label{pro:irrev}
  Let $(X,\RR,\gf)$ be an irreversible RN with reaction energies.
  Then, every futile cycle is a perpetuum mobile. Hence, if
  $(X,\RR,\gf)$ is thermodynamic, then there are no futile cycles.
\end{proposition}
\begin{proof}
  Consider a futile cycle, that is, a flow $\vv > 0$ with $\SM\vv=0$.
  Since all reactions are exergonic, $\vv_r>0$ implies $\gf_r<0$ and
  further $\langle \gf, \vv\rangle < 0$, that is, $\vv$ is a perpetuum
  mobile. Now, if there is a futile cycle and hence a perpetuum mobile,
  then the network is not thermodynamic.
\end{proof}

\subsubsection*{Thermodynamic soundness} 
%\smallskip

We next ask whether a RN $(X,\RR)$ can always be endowed with a vector of
reaction energies $\gf$ such that $(X,\RR,\gf)$ is
thermodynamic. 
\begin{definition}
  A RN $(X,\RR)$ is \emph{thermodynamically sound} if there is a
  vector of reaction energies $\gf$ such that
  $(X,\RR,\gf)$ is a thermodynamic network.
\end{definition}
We note that thermodynamic soundness is a hereditary property of RNs, since
we have seen above that if $(X,\RR,\gf)$ is a thermodynamic network so are
all its subnetworks $(X',\RR',\gf')$.

Again, we first consider purely reversible or irreversible RNs.
\begin{proposition}
  \label{pro:rev} 
  Every reversible RN is thermodynamically sound. 
\end{proposition}
\begin{proof}
  Since $\SM \neq 0$ (the zero matrix), obviously
  $(\ker \SM)^\perp = \im \SM^\trans \neq \{0\}$ (the zero vector), and
  hence there is a non-zero $\gf \in (\ker \SM)^\perp$.
\end{proof}

\begin{theorem} %\label{thm:irrev2} 
  An irreversible RN is thermodynamically sound
  if and only if there are no futile cycles.
\end{theorem}
\begin{proof}
  By Gordan's Theorem (which is in turn a special case of Minty's Lemma
  \cite{Minty:74}, see Appendix B in \cite{Mueller2019}): Either there is a
  negative $\gf \in (\ker \SM)^\perp$ or there is a non-zero, non-positive
  $\vv \in \ker \SM$.  That is, either there is $\gf \ll 0$ with
  $\gf \in (\ker \SM)^\perp$ (the network is thermodyn.\ sound) or there is
  $\vv < 0$ with $\vv \in \ker \SM$; equivalently, there is a futile cycle
  $\vv>0$.
\end{proof}

It is not always obvious from the specification of an artificial chemistry
model whether or not it is thermodynamically sound. As an example, we
consider the artificial chemistry proposed in \cite{Dondi:12}.  It
considers only binary reactions (two educts) that produce two products, aim
to ensure conservation of particle numbers.  In one variant, the networks
only containes irreversible, and thus exergonic reactions. It may produce,
for instance, the following set of reactions:
\begin{equation} 
  \begin{split}
\ce{A + B}&\longrightarrow \ce{C + D} , \\
\ce{A + C}&\longrightarrow \ce{E + B} , \\
\ce{B + D}&\longrightarrow \ce{F + A} , \\
\ce{E + F}&\longrightarrow \ce{A + B} ,
\end{split}
\label{eq:xampl2}
\end{equation}
and assume that all reactions are exergonic.  Their sum corresponds to the
flow $\vv = (1,1,1,1)^\trans \ge 0$ and yields the exergonic composite
reaction
\begin{equation*}
  \begin{split}
    & \ce{2A + 2B + C + D + E + F} \longrightarrow \\
    & \hspace{12ex} \ce{2A + 2B + C + D + E + F} ,
  \end{split}
\end{equation*}
that is, $\mathbf{Sv}=0$.  Thus the model admits a futile cycle composed
entirely of exergonic reactions and hence a perpetuum mobile. Thus it is
not thermodynamically sound. 

\subsubsection*{Mixed Networks}
%\smallskip

In many applications, RNs contain both reversible and irreversible
reactions, $\RR = \RR_\mathrm{rev}\cupdot \RR_\mathrm{irr}$. There are two
interpretations of such models:
\begin{itemize}
\item[(a)] In the (\emph{lax}) sense used above, reversible reactions can be
  associated with arbitrary energies, while irreversible reactions are
  considered exergonic. That is, the reaction energies must satisfy
  $\gf_{r}<0$ for $r\in\RR_{\mathrm{irr}}$.
\item[(b)] In a \emph{strict} sense, the reaction energies assigned to
  irreversible reactions are much more negative than the reaction energies
  of the reversible ones.
  After scaling, one requires \\
  $|\gf_r|\le1$ (that is, $-1 \le \gf_r \le 1$) for
  $r\in\RR_{\mathrm{rev}}$
  and \\
  $|\gf_r|\ge\gamma$ (that is, $\gf_r \le -\gamma$) for
  $r\in\RR_{\mathrm{irr}}$
  and (large) $\gamma>1$. \\
  The intuition is that reactions $r$ with $\gf_r \ge \gamma$ can be
  neglected.
\end{itemize}
The following example shows that thermodynamic soundness differs in the lax
and strict senses.
  
\begin{example}
Consider the following RN
(with reaction energies $\gf$):
\begin{equation} 
  \begin{split}
    1\colon \quad \ce{A}&\longrightarrow \ce{B}, \quad \gf_1=+1, \\
    \bar 1\colon \quad \ce{B}&\longrightarrow \ce{A}, \quad \gf_{\bar 1}=-1, \\
    2\colon \quad \ce{B}&\longrightarrow \ce{C}, \quad \gf_2=-g, \\
    3\colon \quad \ce{C}&\longrightarrow \ce{A}, \quad \gf_3=-g,
  \end{split}
  \label{eq:example2}
\end{equation}
for some $g>0$.
It contains two futile cycles: \\
$\ce{A} \stackrel 1 \to \ce{B} \stackrel {\bar 1} \to \ce{A}$,
$\vv=(1,1,0,0)^\trans$ with $\langle \gf,\vv\rangle = 0$, \\
$\ce{A} \stackrel 1 \to \ce{B} \stackrel 2 \to \ce{C} \stackrel 3 \to
\ce{A}$,
$\vv=(1,0,1,1)^\trans$, $\langle \gf,\vv\rangle = 1-2g$. \\
By setting $g=1/2$, the RN can be made thermodynamic.  (Then the second
futile cycle is not a perpetuum mobile.)

However, the RN in (\ref{eq:example2}) \emph{cannot} be seen as the limit
of a thermodynamic, reversible network
$(\ce{A}\leftrightarrow\ce{B}\leftrightarrow\ce{C}\leftrightarrow\ce{A})$
for large $g$.  Thereby, one considers small $\gf_1,\gf_{\bar 1}$ and large
negative $\gf_2,\gf_3$ (and hence large positive
$\gf_{\bar 2},\gf_{\bar 3}$, that is, negligible reverse reactions
$\bar 2, \bar 3$).  Any such (limit of a) reversible RN contains a
perpetuum mobile (the second futile cycle); equivalently, it is not
thermodynamic.
\end{example}

\begin{definition} \label{def:mixed}
  A mixed network $(X, \RR_\mathrm{rev}\cupdot \RR_\mathrm{irr})$ is
  \emph{thermodynamically sound} if there are reaction energies
  $\gf$ such that $(X,\RR,\gf)$ is thermodynamic and $\gf_r<0$ for
  $r\in\RR_\mathrm{irr}$.
  
  $(X, \RR_\mathrm{rev}\cupdot \RR_\mathrm{irr})$ is \emph{strictly}
  thermodynamically sound if, for all $\gamma>1$, there are reaction
  energies $\gf$ such that $(X,\RR,\gf)$ is thermodynamic, $|\gf_r| \le 1$
  for $r\in\RR_\mathrm{rev}$, and $\gf_r<0$ with $|\gf_r| \ge \gamma$ for
  $r\in\RR_\mathrm{irr}$.
\end{definition}
The scaling condition can also be written in the form
\begin{equation}
  \min_{r\in\RR_{\textrm{irr}}} |\gf_r| \ge 
  \gamma \max_{r\in\RR_{\mathrm{rev}}} |\gf_r| \quad \text{for all }
  \gamma>1.
\end{equation}
A more detailed justification for strict thermodynamic soundness in mixed
networks will be given in the next subsection when considering open RNs.
Here, we focus on the relationship of thermodynamic soundness and futile
cycles.

\begin{theorem} \label{thm:mixed1}
  A mixed RN $(X,\RR_\mathrm{rev}\cupdot \RR_\mathrm{irr})$ is %(laxly)
  thermodynamically sound if and only if there is no irreversible futile
  cycle.
\end{theorem}  
\begin{proof}
  By a ``sign vector version'' of Minty's Lemma: Either there is
  $\gf \in (\ker \SM)^\perp$ with $\gf_r<0$ for $r \in \RR_\mathrm{irr}$
  (the network is thermodynamically sound) or there is a non-zero
  $\vv \in \ker \SM$ with $\vv_r \le 0$ for $r \in \RR_\mathrm{irr}$ and
  $\vv_r = 0$ for $r \in \RR_\mathrm{rev}$; equivalently, there is a futile
  cycle $\vv>0$ with $\supp(\vv) \subseteq \RR_\mathrm{irr}$.
\end{proof}

\begin{theorem} \label{thm:mixed2}
  A mixed RN $(X, \RR_\mathrm{rev}\cupdot \RR_\mathrm{irr})$ is strictly
  thermodynamically sound if and only if no futile cycle %$\vv$
  contains an irreversible reaction.
\end{theorem}
\begin{proof}
  By Minty's Lemma: Let $\gamma>1$.  Either there is
  $\gf \in (\ker \SM)^\perp$ with $\gf_r \in [-1,1]$ for
  $r \in \RR_\mathrm{rev}$ and $\gf_r \in (-\infty,-\gamma]$ for
  $r \in \RR_\mathrm{irr}$ or there is $\vv \in \ker \SM$ with
  \begin{equation} \label{int1} \sum_{r \in \RR_\mathrm{rev}} \vv_r \,
    [-1,1] + \sum_{r \in \RR_\mathrm{irr}} \vv_r \, (-\infty,-\gamma] > 0.
  \end{equation}
  Thereby, a sum of intervals is defined in the obvious way, yielding an
  interval which is positive if each of its elements is positive.  Via
  $\vv\to-\vv$, the interval condition~\eqref{int1} is equivalent to: there
  is $\vv \in \ker \SM$ with
  \begin{equation}
    \label{int2}
    \sum_{r \in \RR_\mathrm{rev}} \vv_r \,
    [-1,1] + \sum_{r \in \RR_\mathrm{irr}} \vv_r \, [\gamma,\infty) > 0 .
  \end{equation}
  As necessary conditions, we find (i) $\vv_{r^*} > 0$ for some
  $r^* \in \RR_\mathrm{irr}$ and (ii) $\vv_{r} \ge 0$ for all
  $r \in \RR_\mathrm{irr}$.  By Lemma~\ref{lem}, (iii) there is an
  equivalent flow with $\vv_{r} \ge 0$ for $r \in \RR_\mathrm{rev}$.  That
  is, there is a futile cycle $\vv>0$ involving an irreversible reaction.
  For $\gamma$ large enough, the necessary conditions are also sufficient
  for the interval condition~\eqref{int2}.
\end{proof}

We may characterize strict thermodynamic soundness for mixed networks also
in geometric terms:
\begin{corollary}
  Let
  $\SM = (\SM_\mathrm{rev} \; \SM_\mathrm{irr}) \in \R^{X \times
    (\RR_\mathrm{rev} \cupdot \RR_\mathrm{irr})}$,
  $L_\mathrm{rev} = \im \SM_\mathrm{rev}$, and
  $C_\mathrm{irr} = \cone \SM_\mathrm{irr}$.  Then, $(X, \RR)$ is strictly
  thermodynamically sound if and only if
  $L_\mathrm{rev} \cap C_\mathrm{irr} = \{0\}$.
\end{corollary}

\subsubsection*{Open (mixed) networks}
%\smallskip

Opening the RN, i.e., adding transport reactions alters the representation
of reaction energies.  We now have to consider chemical potentials
involving concentrations, i.e., we replace the (Gibbs free) energies
$\GF_x$ by $\GF_x + R\,T\ln [x]$, where $[x]$ is the activity of $x$, which
approximately coincides with the concentrations. A reaction $r$ then
proceeds in the forward direction whenenver the chemical potential of the
products is smaller than the chemical potential of the educts, i.e., if
\begin{equation}
  \sum_x s^+_{xr} (\GF_x + R\,T\ln[x]) <
  \sum_x s^-_{xr} (\GF_x + R\,T\ln[x])\,.  
\end{equation}
This condition can be rewritten in terms of the reaction (Gibbs free)
energies and (the logarithm of) the ``reaction quotient'', see e.g.\
\cite{Pekar:05}:
\begin{equation}
  \gf_r < - R\,T \sum_{x\in X} s_{xr}\ln[x]
  \label{eq:openupper}
\end{equation}
The activities $[x]$ for $x\in X$ therefore define an upper bound on the
reaction energy $\gf_r$. In an open system, (internal) concentrations may
be buffered as fixed values or are implictly determined by given influxes
or external concentrations \cite{Polettini:14}.  Given a specification of
the environment, i.e., of the transport fluxes and/or buffered
concentrations, the upper bound in Equ. (\ref{eq:openupper}) can have an
arbitrary value. Thus, if an irreversible reaction in $\RR$ is meant to
proceed forward for all conditions, it must be possible to choose $\gf_r<0$
arbitrarily small, i.e., $|\gf_r|$ arbitrarily large.  This amounts to
requiring that $(X,\RR^p)$ is \emph{strictly} thermodynamically sound. In
many studies of reaction networks, one requires that a reaction proceeds
forward in a given situation, but allows that it proceeds backward in other
situations.  In this (lax) interpretation of irreversibility it is
sufficient to require that $(X,\RR^p)$ is thermodynamically sound, but not
necessarily strictly thermodynamically sound.

In Def.~\ref{def:mixed}, we introduce (strict) thermodynamical soundness in
terms of reaction energies, and in Thms. \ref{thm:mixed1} and
\ref{thm:mixed2}, we characterize it in terms of futile cycles.  In a
corresponding approach \cite{Gorban2011,Gorban2013}, ``extended'' detailed
balance is required for (closed) RNs with irreversible reactions at
thermodynamic equilibrium.  Thereby, activities $[x]$, rate constants
$k_+,k_-$ and equilibrium constants $K$ are explicitly used to formulate
Wegscheider conditions for non-reversible RNs that are limits of reversible
systems.  The characterization of such systems in \cite{Gorban2011} is
equivalent to our results.

\subsubsection*{Reversible Completion}
%\smallskip

As models of chemistry, non-reversible networks are abstractions
that are obtained from reversible thermodynamics networks by omitting the
reverse of reactions that mostly flow into one direction.
\begin{definition}
  \label{def:rc}
  Let $(X,\RR,\gf)$ be a
  thermodynamic RN with
  $\RR=\RR_\mathrm{rev}\cupdot \RR_\mathrm{irr}$.
  The \emph{reversible completion} of
  $(X,\RR,\gf)$ is the RN 
  $(X,\RR^*,\gf^*)$ with
  $\RR^* = \RR_\mathrm{rev}\cupdot \RR_\mathrm{irr} \cupdot \{\bar
  r \mid r\in \RR_\mathrm{irr}\}$ and $\gf^*_r=\gf_r$ for
  $r\in \RR_\mathrm{rev}\cupdot \RR_\mathrm{irr}$ and $\gf^*_{\bar r}=
  -\gf_r$ for $r\in\RR_\mathrm{irr}$.
\end{definition}

\begin{lemma}
  \label{lem:completion}
  If $(X,\RR,\gf)$ is a thermodynamic RN, then its reversible
  completion $(X,\RR^*,\gf^*)$ is also a thermodynamic RN.
\end{lemma}
\begin{proof}
  Let $\bar r\in\RR^*$ be the reverse reaction of $r \in \RR_\mathrm{irr}$.
  By Prop.~\ref{pro:thermodynamic}, for every $r\in\RR$ there is a vector
  $\GF\in\R^X$ satisfying Hess' law. It suffices to show that $\GF$ still
  satisfies Hess' law for $(X,\RR^*)$. By the definition of $\bar r$, its
  reaction energy is
  $\gf^*_{\bar r} = \sum_{x\in X} \GF_x (s^+_{x{\bar r}} - s^-_{x{\bar r}})
  = \sum_{x\in X} \GF_x (s^-_{xr}-s^+_{xr}) = -\gf_r$, as required by
  Def.~\ref{def:rc}. Thus $(X,\RR^*,\gf^*)$ is also thermodynamic.
\end{proof}
The following result is an immediate consequence of Lemma~\ref{lem:completion}.
\begin{proposition}
  If the RN $(X,\RR)$ is thermodynamically sound, then its reversible
  completion is also thermodynamically sound, and the reaction energies
  $\gf$ can be choosen such that $\gf_r<0<\gf_{\bar r}$ for all
  $r\in\RR_\mathrm{irr}$.
\end{proposition}

\subsection{Mass conservation and cornucopias/abysses}
%\smallskip

Thermodynamic soundness is not sufficient to ensure chemical realism.  As
an example, consider the random kinetics model introduced in
\cite{Bigan:13}. It assigns (a randomly chosen) energy $G(x)$ to each
$x\in X$. Each reaction $r$ is defined by randomly picking a set of educts
$e_r^-$ and products $e_r^+$. A possible instance of this model comprises
four compounds with molecular energies $G(\ce{A}) = -5$, $G(\ce{B}) = -5$,
$G(\ce{C}) = -10$, and $G(\ce{X}) = -2$, and two reactions
\begin{equation}
  \ce{A + B} \longrightarrow \ce{C + X}, \quad
  \ce{C} \longrightarrow \ce{A + B}.
  \label{eq:xampl3}
\end{equation}
The first reaction is exergonic with $\gf_1=-2$ and the second has reaction
energy $\gf_2=0$. The composite reaction, obtained as their sum, is
$\ce{A + B} \to \ce{A + B + X}$. Ignoring the effective catalysts \ce{A}
and \ce{B}, the corresponding net reaction is
$\ce{\varnothing} \to \ce{X}$.  In this universe, therefore, it is possible
to spontaneosly create mass in a sequence of exergonic reactions.  Reverting
the signs of the energies reverts the two reactions and thus yields an
exergonic reaction that makes $\ce{X}$ disappear.

We can again describe this situation in terms of flows. Recall that
$[\SM\vv]_x$ is the net production or consumption of species $x$. The
spontaneous creation or annihilation of mass thus corresponds to flows
$\vv>0$ with $\SM\vv>0$ or $\SM\vv<0$, respectively.
\begin{definition}
  Let $(X,\RR)$ be a RN. A flow $\vv>0$ is a \emph{cornucopia} if
  $\SM\vv>0$ and an \emph{abyss} if $\SM\vv<0$.
\end{definition}
Systems with cornucopias or abysses cannot be considered as closed
systems. The proper part of chemical reaction networks therefore must be
free of cornucopias and abysses.

Since in a reversible network any vector $\vv \in \R^\RR$ can be
transformed into an equivalent flow $\mathbf{\tilde v} \ge 0$ (with
$\SM \mathbf{\tilde v} = \SM \vv$), cf.~Lemma~\ref{lem}, we have the
following characterization.
\begin{proposition} %\label{fact:nocornu}
  A reversible RN is free of cornucopias and abysses if and only if there
  is no vector $\vv\in\R^\RR$ such that $\SM\vv>0$.
\end{proposition}

In fact, mass conservation rules out cornucopias and abysses.  More
generally, a \emph{reaction invariant} is a property that does not change
over the course of a chemical reaction
\cite{Horn:72a,Gadewar:01,Flockerzi:07}.  Here, we are only interested in
\emph{linear} reaction invariants, also called \emph{conservation laws}
\cite{Rao:18}, that is, quantitative properties of molecules (such as mass)
whose sum is the same for educts and products.
\begin{definition}
  \label{def:conslaw}
  A \emph{linear reaction invariant} or \emph{conservation law} is a non-zero
  vector $\mm\in \R^X$ that satisfies
  $\sum_{x\in X} \mm_x \, s^+_{xr} = \sum_{x\in X} \mm_x \, s^-_{xr}$ for all
  reactions $r\in\RR$,
  that is, $\mm^\trans\SM=0$. 
\end{definition}
\begin{definition}
  A RN is \emph{conservative} if it has a positive conservation law, that
  is, if there is $\mm\in\R^X$ such that $\mm\gg0$ and $\mm^\trans \SM=0$.
\end{definition}

By definition, a conservative network is free of cornucopias and abysses.
Conversely, by the result below, if a reversible network is not
conservative, then it contains a cornucopia (and an abyss).

\begin{theorem}
  A reversible RN $(X,\RR)$ is free of cornucopias and abysses if and only
  if it is conservative.
\end{theorem}
\begin{proof}
  By Stiemke's Theorem (which is in turn a special case of Minty's Lemma):
  % a general theorem of the alternative, cf.~Appendix B in
  % \cite{Mueller2019})
  Either there is a non-zero, non-negative $\mathbf{n} \in \im \SM$ or there
  is a positive $\mm \in (\im \SM)^\perp = \ker \SM^\trans$.  That is,
  either there is $\vv\in\R^\RR$ with $\mathbf{n} = \SM \vv > 0$
  (corresponding to a cornucopia $\mathbf{\tilde v}>0$) or there is
  $\mm\gg0$ with $\SM^\trans \mm=0$ (as claimed).
\end{proof}

We therefore conclude that every closed chemical RN must have a positive
reaction invariant. This is no longer true if the RN is embedded in an open
system and mass exchange with the environment is allowed.  By construction,
each transport reaction violates at least one of the conservation laws of
the closed system, since $[\mm^\trans\SM]_{r}>0$ if $r$ is import reaction
and $[\mm^\trans\SM]_{r}<0$ if it is an export reaction.  As discussed
e.g.\ in \cite{Rao:18}, opening a RN by adding import or export reactions,
can only reduce the number of conservation laws and cannot introduce
additional constraints. Nevertheless, a RN must be chemically meaningful
when the import and export reactions are turned off. That is, its proper
part $(X,\RR^p)$ must be conservative to ensure that it has a chemical
realization.

\section{Realizations of reaction networks}
%\smallskip

\subsection{Conservation of atoms and moieties}
%\smallskip

Molecules are composed of atoms, which are -- by definition -- preserved in
every chemical reaction. For each atom type $a$, there is a conservation
law that accounts for the number of atoms of type $a$ in each compound~$x$. More precisely, denote by $\mathbf{A}_{ax}\in\N_0$ the number of atoms
of type $a$ in molecule $x$, i.e., the coefficients in the \emph{chemical
  sum formula} $\sum_a \mathbf{A}_{ax} \, a$ for compound $x$.
(Alternatively, we may think of sum formulas as multisets of atoms.)
Conservation of atoms in reaction $r$ therefore becomes
\begin{equation}
  \sum_x \mathbf{A}_{ax} \SM_{xr} = 0 .
\end{equation}
For all reactions and in matrix form, this condition reads
$\mathbf{A}\SM=0$.  Each row of the matrix $\mathbf{A}$ thus is a
non-negative linear reaction invariant, i.e., a non-negative conservation
law.

Conserved \emph{moieties} are groups of atoms that remain intact in all
reactions in which they appear
\cite{Schuster:91,Famili:03,Haraldsdottir:18}. Like atoms, they lead to
non-negative integer conservation laws.

However, (the vectors representing) conserved atoms or moieties need not
span the left kernel of the stoichiometric matrix $\SM$ and need not be
linearly independent.  To see this, consider the following two RNs
comprising a single reaction. For
\begin{equation} 
  \ce{MgCO3 \longrightarrow MgO + CO2}
\end{equation}
with $\SM = (-1, 1, 1)^\trans$, there are only two linearly independent
conservation laws, e.g.\ $(1,1,0)$ and $(1,0,1)$, corresponding to the
\emph{moieties} \ce{MgO} and \ce{CO2}, while the three vectors for the
atomic composition $A_{\ce{Mg}}=(1,1,0)$, $A_{\ce{C}}=(1,0,1)$, and
$A_{\ce{O}}=(3,1,2)$ are linearly dependent. On the other hand, as noted in
\cite{Schuster:91},
\begin{equation}
  \ce{C6H5CH3 + H2 \longrightarrow C6H6 + CH4}
  \label{ce:hydalktol}
\end{equation}
with $\SM=(-1,-1,1,1)^\trans$ has three conservation laws but only two atom
types, which correspond to the conservation laws $A_{\ce{C}}=(7,0,6,1)$ and
$A_{\ce{H}}=(8,2,6,4)$. E.g.\ the phenyl-moiety $M_{ph}=(1,0,1,0)$ or the
methyl-moiety $M_{\ce{CH4}}=(1,0,0,1)$ form the missing third, linearly
independent conservation law. The latter example also shows that atom
conservation relations are not necessarily support-minimal among the
non-negative integer left-kernel vectors of~$\SM$. In fact, also $(0,1,1,0)$
and $(0,1,0,1)$ are left-kernel vectors of~$\SM$, the chemical
interpretation of which is less obvious.

These examples show that key chemical properties such as atom conservation
or conservation of moieties are \emph{not} encoded in the stoichiometric
matrix $\SM$. In other words, two RNs can be isomorphic as hypergraphs but
describe reactions between sets of compounds that are not isomorphic in
terms of their sum formulas. For example, $\SM=(-1,-1,1,1)^\trans$ is
realized by the hydroalkylation of toluene in Equ.~(\ref{ce:hydalktol}),
but also by the inorganic reaction
\begin{equation}
  \ce{ MgO + H2SO4 \longrightarrow MgSO4 + H2O } ,
\end{equation}
having four atom conservation laws, $A_{\ce{Mg}}=(1,0,1,0)$,
$A_{\ce{O}}=(1,4,4,1)$, $A_{\ce{H}}=(0,2,0,2)$, $A_{\ce{S}}=(0,1,1,0)$, and
three moiety convervation laws, e.g.  $M_{\ce{MgO}}= (1,0,1,0)$,
$M_{\ce{H2O}}= (0,1,0,1)$, and $M_{\ce{SO3}}= (0,1,1,0)$.

\emph{``Semi-positive'' conservation laws} \cite{Schuster:91,Schuster:95}
of a RN are the non-zero elements of the polyhedral cone
\begin{equation}
  K(\SM) = \left\{ \yy\in \R^{X} \mid \yy\SM=0,
    \, \yy \ge 0 \right\},
\end{equation}
the non-negative left-kernel of $\SM$.  Thereby, $K(\SM)$ is an s-cone as
defined in \cite{MuellerRegensburger2016}, given by a subspace (here:
$\ker \SM^\trans$) and non-negativity conditions.  Since the s-cone $K(\SM)$
is contained in the non-negative orthant, its extreme (non-decomposable)
vectors agree with its support-minimal vectors.  Further, since $\SM$ is an
integer matrix, all extreme vectors of $K(\SM)$ are positive real multiples
of integer vectors.

All potential \emph{moiety conservation laws} (MCLs) \cite{DeMartino:14}
for a given stoichiometric matrix $\SM$ (but unknown atomic composition) are
non-zero, integer elements of $K(\SM)$, i.e., elements of the set
\begin{equation}
  \mathcal{K}(\SM) =
  \left\{ \yy\in \N_0^{X} \mid \yy\SM=0 \right\} \setminus \{0\} .
\end{equation}
Clearly, $\mathcal{K}(\SM)$ contains the integer extreme vectors of
$K(\SM)$.  Ultimately, one is interested in {\em minimal} MCLs, i.e.,
minimal elements of $\mathcal{K}(\SM)$, cf.~\cite{Graver1975}.  (Minimal
vectors are called maximal in \cite{Schuster:95}.)
\begin{definition}
  A vector $\yy\in\mathcal{K}(\SM)$ is \emph{minimal} if there is no
  $\mathbf{y'}\in \mathcal{K}(\SM)$ such that $\mathbf{y'}<\yy$.
\end{definition}
In fact, integer minimality and integer non-decompos\-ability are
equivalent.
\begin{proposition}
  \label{pro:intmin}
  Let $\yy\in\mathcal{K}(\SM)$.  The following statements are equivalent:
  \begin{itemize}
  \item[1.] $\yy$ is minimal.
  \item[2.] There are no two $\yy', \yy''\in\mathcal{K}(\SM)$
    such that $\yy=c'\yy'+c''\yy''$ with
    $c',c''\in\N$.
  \end{itemize}
\end{proposition}
\begin{proof}
  Suppose $\yy'<\yy$. Then $\yy=1\cdot(\yy-\yy')+1\cdot \yy'$.
  Conversely, suppose $\yy=c'\yy'+c''\yy''$. Then $\yy', \yy'' < \yy$.
\end{proof}
Most importantly, the minimal MCLs generate all MCLs.

\begin{theorem}
  Every element of $\mathcal{K}(\SM)$ is a finite integer linear
  combination of minimal elements of $\mathcal{K}(\SM)$.
\end{theorem}
\begin{proof}
  By Noetherian induction on the partial order $<$ on $\N^X_0$ and
  Proposition~\ref{pro:intmin}.
\end{proof}
Knowing all minimal MCLs allows to represent the compounds $X$ of a RN
$(X,\RR)$ in a minimal (most coarse-grained) way.
\begin{definition}
  The minimal moiety representation (short: \emph{mm-representation}) of a
  conservative RN $(X,\RR)$ is the matrix
  $\mathbf{M}\in\N_0^{\mathcal{M} \times X}$, where the rows of
  $\mathbf{M}$ are the minimal MCLs, and $\mathcal{M}$ is the corresponding
  set of abstract moieties.
\end{definition}
For example, consider the abstract chemical reaction
\begin{equation} \label{ex:ABC}
\ce{A + B $\longrightarrow$ 2C}
\end{equation}
with $\SM = (-1,-1,2)^\trans$.  There are three minimal MCLs denoted by the
abstract moieties $\mathcal{M} = \{ \ce{X,Y,Z} \}$: on the one hand,
$M_{\ce{X}}=(2,0,1)$ and $M_{\ce{Y}}=(0,2,1)$, which are (minimal) extreme
vectors of $K(\SM)$, on the other hand, $M_{\ce{Z}}=(1,1,1)$, which is
minimal, but not extreme.  Hence, the mm-representation is given by
\begin{equation} 
  \mathbf{M} = \begin{pmatrix} 2&0&1\\ 0&2&1 \\ 1&1&1 \end{pmatrix},
\end{equation}
and the reaction \eqref{ex:ABC} can be represented as
\begin{equation} 
  \ce{X_2Z + Y_2Z $\longrightarrow$ 2 XYZ} .
\end{equation}

By definition, $\im \mathbf{M}^\trans \subseteq \ker \SM^\trans$.  In fact,
$\im \mathbf{M}^\trans = \ker \SM^\trans$, and hence there is an obvious
lower bound for the number of minimal MCLs.
\begin{lemma} \label{lem:dim} Let
  $\mathbf{M}\in\N_0^{\mathcal{M} \times X}$ be the mm-representation of a
  conservative RN $(X,\RR)$ with stoichiometric matrix~$\SM$. Then,
  $\im \mathbf{M}^\trans = \ker \SM^\trans$ and hence
  $|\mathcal{M}|\ge\dim \ker \SM^\trans$.
\end{lemma}
\begin{proof}
  Since the left kernel of $\SM$ and hence $K(\SM)$ contain a positive
  vector, we have $\dim K(\SM) = \dim \ker \SM^\trans \eqqcolon d$.  Hence,
  (the extreme vectors of) $K(\SM)$ and therefore also (the corresponding
  minimal integer vectors of) $\mathcal{K}(\SM)$ generate
  $\ker \SM^\trans$, that is, $\im \mathbf{M}^\trans = \ker \SM^\trans$.
  Hence, the number of minimal MCLs is greater equal $d$, that is,
  $|\mathcal{M}|\ge\dim \ker \SM^\trans$.
\end{proof}

By instantiating the abstract moieties $\{ \ce{X},\ce{Y},\ce{Z} \}$ with
sum formulas (multisets of atoms), every chemical realization of the
reaction can be obtained.  In general, we define an instance as follows.
\begin{definition} \label{def:sfins}
  A sum formula instance (short: \emph{sf-instance}) of a %closed
  RN $(X,\RR)$ with stoichiometric matrix~$\SM$ is a matrix
  $\mathbf{A}\in\N_0^{\mathcal{A} \times X}$ for some non-empty,
  finite set $\mathcal{A}$ of ``atoms'' such that
  \begin{itemize}
  \item[(i)] each column of $\mathbf{A}$ is non-zero, and
  \item[(ii)] $\mathbf{A} \SM = 0$.
  \end{itemize}
\end{definition}
Def.~\ref{def:sfins} in particular allows that $\mathbf{A}$ comprises a
single row. By condition (i), this row vector is a strictly positive
conservation law, which, as a linear combination of MCLs, may be chosen to
be integer valued. Conversely, if $(X,\RR)$ admits an sf-instance, then the
column-sum $\mathbf{m}= \mathbf{1}^\trans\mathbf{A}\in \ker\SM^\trans$ is a
strictly positive integer conservation law and thus in particular an
sf-instance with $|\mathcal{A}|=1$.  Taken together, we have shown the
following existence result.
\begin{proposition} \label{prop:sf-instance}
  A RN $(X,\RR)$ admits an \emph{sf-instance} if and only if it is
  conservative.
\end{proposition}
The entry $\mathbf{m}_x$ of $\mathbf{m}$ can be interpreted as the total
number of atoms in compound $x\in X$. In \cite{Doty:18}, a RN is called
\emph{primitive atomic} if each reaction preserves the total number of
atoms. Thus a RN is primitive atomic if and only if it is conservative,
cf.~\cite{Doty:18}.

\subsection{Isomers and Sum Formula Realizations}
%\smallskip

In order to gain a better understanding of sf-instances for a RN $(X,\RR)$,
we consider net reactions of the form \ce{X\to Y} in the reversible
completion of $(X,\RR)$. That is, we ask whether it is possible, in
principle, to convert \ce{X} into \ce{Y}, irrespective of whether the
conversion is thermodynamically favorable. From a chemical perspective, if
such a \emph{net isomerization reaction} exists, then \ce{X} and \ce{Y}
must be compositional isomers. These will play a key role in our discussion
of realizations of $(X,\RR)$ in terms of sf-instances.

Before we proceed, we first give a more formal account of net isomerization
reactions. Recall that a net reaction derives from an overall reaction,
which in turn is specified by an integer hyperflow. Instead of working
explicitly in the reversible completion, we may instead consider vectors
$\vv\in\mathbb{Z}^{\RR}$ with negative entries $\vv_r<0$, representing the
reverse of irreversible reactions $r\in\RR$.
\begin{definition}
  \label{def:obiso}
  Let $(X,\RR)$ be a RN with stoichiometric matrix $\SM$.  A vector
  $\vv\in\mathbb{Z}^{\RR}$, satisfying
  $k\coloneqq -[\SM\vv]_x = [\SM\vv]_y\in\N$ for some $x,y \in X$ and
  $[\SM\vv]_z=0$ for all $z\in X\setminus\{x,y\}$, specifies a \emph{net
    isomerization reaction} $k\,x\to k\,y$.  Two (distinct) compounds
  $x,y\in X$ are \emph{obligatory isomers} if $(X,\RR)$ admits a net
  isomerization reaction $k\,x\to k\,y$.  We write $x\rightleftharpoons y$
  if $x=y$ or $x$ and $y$ are obligatory isomers.
\end{definition}

\begin{proposition}
  \label{prop:equirel}
  The binary relation $x\rightleftharpoons y$ introduced in
  Def.~\ref{def:obiso} is an equivalence relation.
\end{proposition}
\begin{proof}
  By definition, $\rightleftharpoons$ is reflexive.  If $\vv$ specifies the
  net isomerization reaction $k\,x\to k\,y$, then $-\vv$ specifies
  $k\,y\to k\,x$, and thus $\rightleftharpoons$ is symmetric. To verify
  transitivity, suppose $x\rightleftharpoons y$ and
  $y\rightleftharpoons z$, i.e., there are vectors $\vv^1$ and $\vv^2$ that
  specify the net isomerization reactions $p\,x \to p\,y$ and
  $q\,y \to q\,z$. Then $\vv = q\vv_1 + p\vv_2$ satisfies $[\SM\vv]_x=-pq$,
  $[\SM\vv]_z=pq$, $[\SM\vv]_y=0$, and $[\SM\vv]_u=0$ for all
  $u\in X\setminus\{x,y,z\}$, and thus specifies the net isomerization
  reaction $(pq)\,x \to (pq)\,z$. Thus, $\rightleftharpoons$ is transitive.
\end{proof}

The intuition is to define a \emph{sum formula realization} of a RN as a matrix
$\mathbf{A}$ that (i) is an sf-instance of the RN and (ii) assigns
different atomic compositions to $x$ and $y$ whenever
$x\not\rightleftharpoons y$, that is, whenever $x$ and $y$ are not isomers.
In the following, we will see that such a definition both ensures chemical
realism and leads to a useful mathematical description. The next result
relates net isomerization reactions to the structure of $\ker\SM^\trans$
(and ultimately to compositional isomers as given by MCLs and
sf-instances).

\begin{theorem}
  \label{thm:iso<->kernel}
  Let $(X,\RR)$ be a RN with stoichiometric matrix $\SM$. 
  Then $x\rightleftharpoons y$ if and only if
  $\mathbf{m}_x=\mathbf{m}_y$ for all $\mathbf{m}\in\ker\SM^\trans$.
\end{theorem}
\begin{proof}
  First suppose $x\rightleftharpoons y$. Then either $x=y$ (in which case
  the assertion is trivially true) or there is a net isomerization reaction
  $k\,x\to k\,y$ specified by the vector $\vv$. Let
  $\mathbf{m}\in\ker\SM^\trans$. By the definition of $\vv$, we have
  $0 = \mathbf{m}^\trans \SM \vv = \sum_{z\in X} \mathbf{m}_z[\SM\vv]_z =
  \mathbf{m}_x[\SM\vv]_x +\mathbf{m}_y[\SM\vv]_y =
  (\mathbf{m}_x-\mathbf{m}_y)[\SM\vv]_x$ and $[\SM\vv]_x\ne 0$.  Hence,
  $\mathbf{m}_x=\mathbf{m}_y$.

  Now suppose $\mathbf{m}_x=\mathbf{m}_y$ for all
  $\mathbf{m}\in\ker\SM^\trans$ and consider the vector
  $\mathbf{w}\in \mathbb{Z}^{X}$ with $\mathbf{w}_x=-1$, $\mathbf{w}_y=1$, and
  $\mathbf{w}_z=0$ for all $z\in X\setminus\{x,y\}$.  Clearly,
  $\langle \mathbf{m},\mathbf{w}\rangle=0$ for all
  $\mathbf{m}\in\ker\SM^\trans$, that is,
  $\mathbf{w}\in (\ker\SM^\trans)^\perp = \im \SM$.  Thus there is
  $\vv\in\R^{\RR}$ such that $\mathbf{w}=\SM\vv$. Since
  $\SM\in\mathbb{Z}^{X\times\RR}$, the solution $\vv$ of this linear
  equation is rational. Writing $\lcd(\vv)$ for the least common denomintor
  of the entries in $\vv$, we obtain the integer vector
  $\lcd(\vv)\, \vv\in \mathbb{Z}^{\RR}$, specifying the net isometrization
  reaction $\lcd(\vv)\, x\to \lcd(\vv)\, y$.  By definition,
  $x\rightleftharpoons y$.
\end{proof}

The proof of Thm.~\ref{thm:iso<->kernel} also provides a simple algorithm
to compute integer hyperflows $\vv$ that specify net isomerization
reactions and to identify the obligatory isomers: For each pair $x,y\in X$,
construct $\mathbf{w}$ with $\mathbf{w}_x=-1$ and $\mathbf{w}_y=1$ being
the only non-zero entries and solve the linear equation $\SM\vv=\mathbf{w}$. We
have $x\rightleftharpoons y$ if and only if a solution exists, in which
case the desired integer hyperflow is $\lcd(\vv)\,\vv$. 

\begin{figure*}
  \begin{center}
%    \begin{minipage}{0.6\textwidth}
      \includegraphics[width=0.6\textwidth]{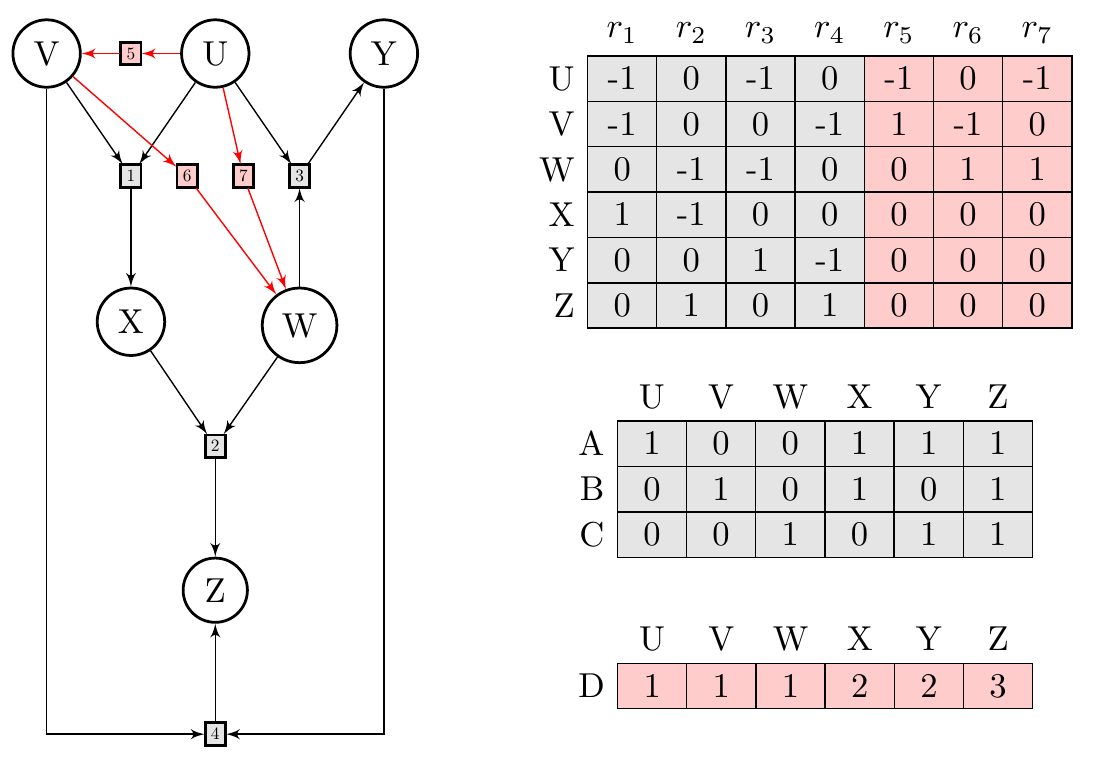}
%    \end{minipage}
%    \quad
%    \begin{minipage}{0.3\textwidth}
      \caption{Reaction network (left) and stoichiometric matrix $\SM$ (top
        right) showing reactions $r_1$-$r_4$, Equ.(\ref{eq:first}), in gray
        and the isomerization reactions $r_5$-$r_7$, Equ.(\ref{eq:second})
        in light red. For the basic system (gray) we have
        $\dim\ker\SM^\trans=3$. The three MCLs are shown below $\SM$. In
        the full system, $r_1$-$r_7$, we have $\dim\ker\SM^\trans=1$ with
        the unique MCL shown at the bottom right. In the full system
        \ce{U}, \ce{V}, and \ce{W} form obligatory isomers of the monomer
        \ce{D}.  Similarly, \ce{X} and \ce{Y} are also obligatory isomers
        composed of two D units, while \ce{Z} is a trimer of \ce{D} units.
        The vector $\vv=(-1,0,1,0,0,1,0)$ is represented by the composite
        reaction %\newline 
        \ce{X + (U + W) + V \to (U + V) + Y + W} %\newline
        and specifies the net isometrization reaction \ce{ X \to Y }.}
      \label{fig:obligisos}
%    \end{minipage}
  \end{center}
\end{figure*}

We next show that obligatory isomers cannot be distinguished by
sf-instances, and conversely, compounds that are not obligatory isomers are
distinguished by certain sf-instances.
\begin{theorem}
  \label{thm:sfreal}
  Let $(X,\RR)$ be a RN with stoichiometric matrix $\SM$
  and $\mathbf{A}\in\N_0^{\mathcal{A}\times X}$ be an
  sf-instance. If $\im \mathbf{A}^\trans = \ker \SM^\trans$, then the
  following statements are equivalent:
\begin{itemize}
\item[(i)] $x,y \in X$ are obligatory isomers;
\item[(ii)] $\mathbf{A}_{ax}=\mathbf{A}_{ay}$ for all $a\in\mathcal{A}$.
\end{itemize}
If $\im \mathbf{A}^\trans \subseteq \ker \SM^\trans$, then (i) implies (ii).
\end{theorem}
\begin{proof}
  Let $x,y \in X$ be distinct.  On the one hand, by
  Theorem~\ref{thm:iso<->kernel}, statement (i) is equivalent to
  $\mathbf{m}_x=\mathbf{m}_y$ for all $\mathbf{m} \in \ker \SM^\trans$.  On
  the other hand, statement (ii) is equivalent to
  $\mathbf{m}_x=\mathbf{m}_y$ for all
  $\mathbf{m} \in \im \mathbf{A}^\trans$.  If
  $\im \mathbf{A}^\trans = \ker \SM^\trans$, then (i) and (ii) are equivalent.
  If $\im \mathbf{A}^\trans \subseteq \ker \SM^\trans$, that is, if the
  rows of $\mathbf{A}$ are elements of $\ker \SM^\trans$, then (i) implies
  (ii).
\end{proof}

Any sf-instance $\mathbf{A}$ whose rows span $\ker\SM^\trans$ not only
identifies obligatory isomers, but also assigns distinct sum formulas to
any distinct compounds $x,y\in X$ that are not obligatory isomers. In this
case, there is at least one row (corresponding to atom $a$) for which
$\mathbf{A}_{ax}\ne\mathbf{A}_{ay}$.  This provides the formal
justification for a mathematical definition of sum formula realizations.
\begin{definition}
  \label{def:sfrel}
  A sum formula realization (short: \emph{sf-realization}) of a RN
  $(X,\RR)$ with stoichiometric matrix $\SM$ is a matrix
  $\mathbf{A}\in\N_0^{\mathcal{A} \times X}$ for some non-empty, finite set
  $\mathcal{A}$ of ``atoms'' such that
  \begin{itemize}
  \item[(i)] each column of $\mathbf{A}$ is non-zero and
  \item[(ii)] $\im \mathbf{A}^\trans = \ker\SM^\trans$.
  \end{itemize}
\end{definition}
As an illustration, consider the RN
\begin{equation}
  \label{eq:first}
  \begin{split}
    \ce{U + V $\longrightarrow$ X}, &\quad \ce{U + W $\longrightarrow$ Y}, \\
    \ce{X + W $\longrightarrow$ Z}, &\quad \ce{Y + V $\longrightarrow$ Z},
  \end{split}
\end{equation}
depicted on the left side of Fig.~\ref{fig:obligisos}. The RN can be
instantiated by the sum formulas $\ce{U} = \ce{A}$, $\ce{V}=\ce{B}$,
$\ce{W}=\ce{C}$, $\ce{X}=\ce{AB}$, $\ce{Y}=\ce{AC}$, $\ce{Z} = \ce{ABC}$.
The corresponding matrix~$\mathbf{A}$ (middle right in
Fig.~\ref{fig:obligisos}) is not only an sf-instance, its rows also span
$\ker\SM^\trans$, and hence it is an sf-realization. (In fact, it is also
the mm-representation.)  A ``reduced representation'' can be obtained by
\emph{assuming} that \ce{U}, \ce{V}, and \ce{W} are compositional isomers
corresponding to the same moiety~\ce{D}, that is,
$\ce{U} = \ce{V} = \ce{W} = \ce{D}$.  As a consequence, \ce{X} and \ce{Y}
are also compositional isomers, $\ce{X} = \ce{Y} = \ce{D2}$, and further
$\ce{Z} = \ce{D3}$.  The corresponding matrix $\mathbf{A'}$ still defines
an sf-instance, but its rows do not span $\ker\SM^\trans$.  Now consider an
extension of the RN in Equ.~(\ref{eq:first}), by adding three
isometrization reactions,
\begin{equation}
  \label{eq:second}
  \ce{U $\longrightarrow$ V}, \quad \ce{V $\longrightarrow$ W}, \quad
  \ce{U $\longrightarrow$ W}.
\end{equation}
In the extended network given by Equ.~(\ref{eq:first}) and
Equ.~(\ref{eq:second}), we have $\dim\ker\SM^\trans=1$, and thus there is a
unique MCL. The reactions in Equ.~(\ref{eq:second}) now \emph{enforce} that
\ce{U}, \ce{V}, and \ce{W} are compositional isomers and thus correspond to
the same moiety~\ce{D}.  This coincides with the ``reduced representation''
$\mathbf{A'}$ for the RN in Equ.~(\ref{eq:first}). The distinction is that,
for the RN of Equ.~(\ref{eq:first}), we may (but do not have to) assume
that \ce{U}, \ce{V}, and \ce{W} are isomers, whereas in the extended
network, no other interpretation is possible.

Finally, we characterize RNs that admit an sf-realization.
\begin{proposition}
  \label{prop:sf}
  A RN $(X,\RR)$ admits an \emph{sf-realization} if and only if it is
  conservative.
\end{proposition}
\begin{proof}
  Suppose $(X,\RR)$ admits an sf-realization, which, in particular, is an
  sf-instance.  By Prop.~\ref{prop:sf-instance}, $(X,\RR)$ is conservative.
  Conversely, suppose $(X,\RR)$ is conservative.  By definition, the
  mm-representation is an sf-instance, and by Lemma~\ref{lem:dim}, it is an
  sf-realization.
\end{proof}
  
Obligatory isomers put some restriction on sf-instances.  Still, there is
surprising freedom for sf-realizations.  We say that two sf-realizations
$\mathbf{A}$ and $\mathbf{A'}$ are equivalent, $\mathbf{A}\sim\mathbf{A'}$,
if there are integers $p,q\in\N$ such that $p\mathbf{A}=q\mathbf{A'}$. One
easily checks that $\sim$ is an equivalence relation.  If
$\dim\ker\SM^\trans=1$, then all $\mathbf{m}\in \dim\ker\SM^\trans$ are
multiples of the unique minimal MCL.  All sum formulas are then of the form
$\ce{D}_k$, and thus we can think of compounds simply as integers $k\in\N$.
Every reaction thus can be written in the form
$\sum_k s_{kr}^- \ce{D}_k \to \sum_k s_{kr}^+ \ce{D}_k$ with
$\sum_k (s_{kr}^+ - s_{kr}^-)k = 0$.
An example of practical interest is the rearrangement chemistry of
carbohydrates, found in metabolic networks such as the pentose phosphate
pathway (PPP) or the non-oxidative part of the Calvin–Benson–Bassham (CBB)
cycle in the dark phase of photosynthesis. Carbohydrates may be seen as a
``polymers'' of formaldehyd units and can therefore be written as \ce{D_k =
  (CH2O)_k}. The PPP interconverts pentoses (e.g.\ ribose) and hexoses
(such as glucose), in an atom-economic (no waste) rearrangement network
possessing the overall reaction \ce{6 (CH2O)_5 <=> 5 (CH2O)_6}. In a
similar fashion five 3-phosphoglycerates are reconfigured via carbohydrate
chemistry into three ribulose-5-phosphate which results in the overall
reaction of \ce{5 (CH2O)_3 -> 3 (CH2O)_5} if focusing on the sugar
component. Carbohydrate reaction chemistry is particularly well-suited for
the implementation of isomerization networks, and the logic and structure
of the design space of alternative networks implementing the same overall
reaction has been explored using mathematical and computational models
\cite{Meliendez-Hevia:1985,Andersen:19a}.

For $\dim\ker\SM^\trans > 1$, there is an infinite set of sf-realizations
that are pairwisely inequivalent. To see this, construct matrices
$\mathbf{A}_{t} = (t_1\mathbf{y}^1,t_2\mathbf{y}^2,\dots,t_k\mathbf{y}^k)^\trans$
from $k=\dim\ker\SM^\trans>1$ linearly independent (minimal) MCLs $\yy^i$
and with $t\in\N^k$.  Clearly, every such matrix $\mathbf{A}_{t}$ is an
sf-realization. Furthermore, $\mathbf{A}_{t}\sim \mathbf{A}_{t'}$ if and
only if there are $p,q\in\N$ such that $p{t}= q{t'}$. Hence
$\mathbf{A}_{t}\not\sim \mathbf{A}_{t'}$ if there are two distinct indices
$1\le i<j\le k$ such that $t_i / t'_i \ne t_j / t'_j$.  Clearly, there is
an infinite set $T \subseteq \N^k$ of integer vectors such that this
inequality is satisfied for all distinct ${t},{t'}\in T$. For instance, one
may choose distinct primes for all entries of ${t} \in T$.  Thus there are
infinitely many pairwisely inequivalent sf-realizations.  Furthermore, the
choice of the (minimal) MCLs is not unique, in general, allowing additional
freedom for sf-realizations. Finally, one may produce more complex
sf-realizations by appending additional rows to $\mathbf{A}$ that are
linear combinations of the basis vectors.  Therefore we have the following
result.
\begin{proposition}
  Let $(X,\RR)$ be a conservative RN with stoichiometric matrix $\SM$.  If
  $\dim\ker\SM^\trans>1$, then there are infinitely many in-equivalent
  sf-realizations of $(X,\RR)$.
  \label{prop:infty}
\end{proposition}

\subsection{Structural Formula Realizations} 
%\smallskip

A structural formula represents a chemical species as a (connected)
molecular graph, whose vertices are labeled by atom types and edges refer
to chemical bonds. \emph{Lewis structures} \cite{Lewis:1916} are equivalent
to vertex-labeled multigraphs in which each bonding electron pair is
represented as an individual edge, and each non-bonding electron pair as a
loop. In particular, double or triple bonds are shown as two or three
parallel edges. The educt and product complexes $r^-$ and $r^+$ of a
reaction $r$ can then be represented as the disjoint unions of the educt
and product graphs, respectively. A chemical reaction is a graph
transformation that converts the educt graph into the product graph such
that vertices and their labels are preserved
\cite{Benkoe:03b,Rossello:05}. Only the bonds are rearranged. Since
electrons are conserved, and each edge or loop accounts for two electrons,
any reaction must preserve the sum of vertex degrees and thus the
  number of edges. Fig.~\ref{fig:H2SO4} shows an example.

\begin{figure}
  \begin{center}
    \includegraphics[width=0.9\columnwidth]{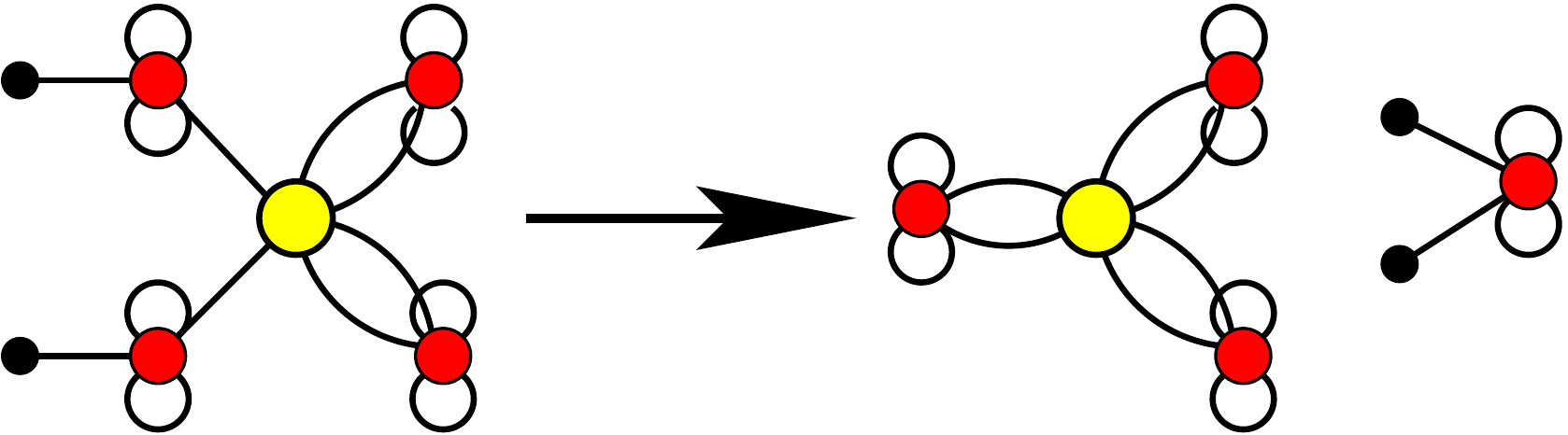}
  \end{center}
  \caption{Multigraph representation for the reaction \ce{H2SO4 -> SO3 +
      H2O}. Atoms shown in color: \ce{H} black, \ce{O} red, \ce{S} in
    yellow. Non-bonding electron pairs are represented by loops, double
    bonds by two parallel edges.}
  \label{fig:H2SO4}
\end{figure}

This idea can be generalized to sf-realizations in which ``atoms'' are
viewed as moieties. We may then interpret the vertices of a multigraph as
``fragments'' of species that are endowed with a certain number of
``valencies'' or ``half bonds''. These must be ``saturated'' by binding to
free valencies of other moieties or they must be used to form internal
bonds within a moiety. Each atom type/moiety has a fixed degree given as
the number of halfbonds (binding to other moieties or forming loops).
Correspondingly, the degree of a vertex $u$ in a multigraph is defined as
the number of edges that connect $u$ with other vertices plus twice the
number of loops.  A reaction thus preserves electrons if and only if its
only effect is to rearrange the bonds in the multigraph.

\begin{definition} \label{def:sf1}
  Let $\mathcal{A}$ be a non-empty, finite set,
  $\val\colon\mathcal{A}\to\N$ be an arbitrary function, and
  $\sum_{a\in\mathcal{A}} n_a \, a$ be a sum formula.  A multigraph
  $\Gamma= (V,E,\alpha)$ with loops and vertex coloring
  $\alpha\colon V\to\mathcal{A}$ is a corresponding \emph{structural
    formula} if it satisfies the following conditions:
  \begin{itemize}
  \item[(i)] Each vertex $u\in V$ corresponds to a moiety $\alpha(u)$, in
    particular, $|\{u\in V\colon\alpha(u)=a\}|=n_a$.
  \item[(ii)] $d(u)=\val(\alpha(u))$ for all $u\in V$, i.e., the vertex
    degree of $u$ is given by the corresponding moiety.
  \item[(iii)] $\Gamma$ is connected.
  \end{itemize}
  \label{def:structform}
\end{definition}%

The structural formulas specified in Def.~\ref{def:structform} do not
cover all Lewis structures. In particular, neither explicit charges
nor unpaired electrons are covered. While these are important from a
chemical perspective, we shall see below that such extensions are not
needed for our purposes since the straightforward multigraphs in
Def.~\ref{def:structform} already provide sufficient freedom to obtain
representations for all conservative RNs. 

\begin{definition} \label{def:sf2}
  Let $(X,\RR)$ be a RN, $\mathcal{A}$ be a non-empty, finite set, and
  $\val\colon\mathcal{A}\to\N$ be an arbitrary function.  A \emph{Lewis
    instance} is an assignment of vertex-colored multigraphs
  $\Gamma_x=(V_x,E_x,\alpha_x)$ to all $x\in X$ such that
  \begin{itemize}
  \item[(i)] 
    vertex degrees satisfy $d(u)=\val(\alpha_x(u))$, for
    all $u\in V_x$ and $x\in X$, and
  \item[(ii)] the corresponding matrix
    $\mathbf{A} \in \N_0^{\mathcal{A}\times X}$ defined by
    $\mathbf{A}_{ax} = |\{ u\in V_x\colon \alpha_x(u)=a\}|$ is an
    sf-instance.
  \end{itemize}
  Furthermore, $x\mapsto\Gamma_x$ is a \emph{Lewis realization}
   if $\mathbf{A}$ is an sf-realization. 
\end{definition}

Clearly, every Lewis realization has a corresponding sf-realization.  Given
an sf-realization, we therefore ask when there is a corresponding Lewis
realization.  By Def.~\ref{def:sf1} and~\ref{def:sf2}, we have the
following result.
\begin{lemma}
  \label{lem:suff}
  A RN $(X,\RR)$ has a Lewis realization with corresponding sf-realization
  $\mathbf{A}\in\N_0^{\mathcal{A}\times X}$ for some non-empty, finite set
  $\mathcal{A}$, if and only if there is a function
  $\val\colon\mathcal{A}\to\N$ such that for the sum formula
  $\sum_{a\in\mathcal{A}} \mathbf{A}_{ax} \, a$ for $x\in X$) there is a
  corresponding structural formula $\Gamma_x$.
\end{lemma}
\begin{proof}
  For the 'if' part, let $\sum_{x\in\mathcal{A}} \mathbf{A}_{ax} \,a$ be
  the sum formula for $x\in X$. By assumption, there exists a
  vertex-colored multigraph $\Gamma_x=(V_x,E_x,\alpha_x)$ for $x$ such that
  (i) vertex degrees satisfy $d(u)=\val(\alpha_x(u))$ and (ii) the
  corresponding matrix equals the sf-realization $\mathbf{A}$.
  Analogously, for the 'only if' part.
\end{proof}

The appeal of this characterization is that it does not use any properties
of the RN $(X,\RR)$, at all. In fact, it is easy to see that such a
representation always exists.
\begin{lemma}
  \label{lem:cycle}
  Let $\mathcal{A}$ be a nonempty, finite set and
  $\sum_{a\in\mathcal{A}} n_a \, a$ be a sum formula. 
  Then, there exists a corresponding structural formula with $\val(a)=2$
  for all $a\in\mathcal{A}$.
\end{lemma}
\begin{proof}
  If the sum formula is given by $n_a=1$ and $n_{a'}=0$ for all
  $a'\in\mathcal{A}\setminus\{a\}$, i.e., if it is single moiety, then the
  corresponding structural formula is a single vertex with color $a$ and a
  loop. Otherwise, arrange the $|V|=\sum_a n_a$ vertices, of which exactly
  $n_a$ are colored by $a$, in a cycle and connect the vertices along the
  cycle. Then every vertex $u$ satisfies $d(u)=\val(\alpha(u))=2$ and the
  graph is connected.
\end{proof}
The result extends to any constant function $\val(a)=2k$ (with $k \in \N$) by adding
  loops $k-1$ loops to each vertex. As an immediate consequence of
  Lem.~\ref{lem:suff} and~\ref{lem:cycle}, we have the following result.
\begin{proposition}
  $(X,\RR)$ has a Lewis realization if and only if has an sf-realization.
\end{proposition}
Using Prop.~\ref{prop:sf}, we can characterize RNs that admit a Lewis
realiztion.
\begin{proposition}
  \label{pro:structform}
  A RN $(X,\RR)$ admits a Lewis realization if and only if it is
  conservative.
\end{proposition}

\begin{figure*}
  \begin{center}
    \includegraphics[width=0.9\textwidth]{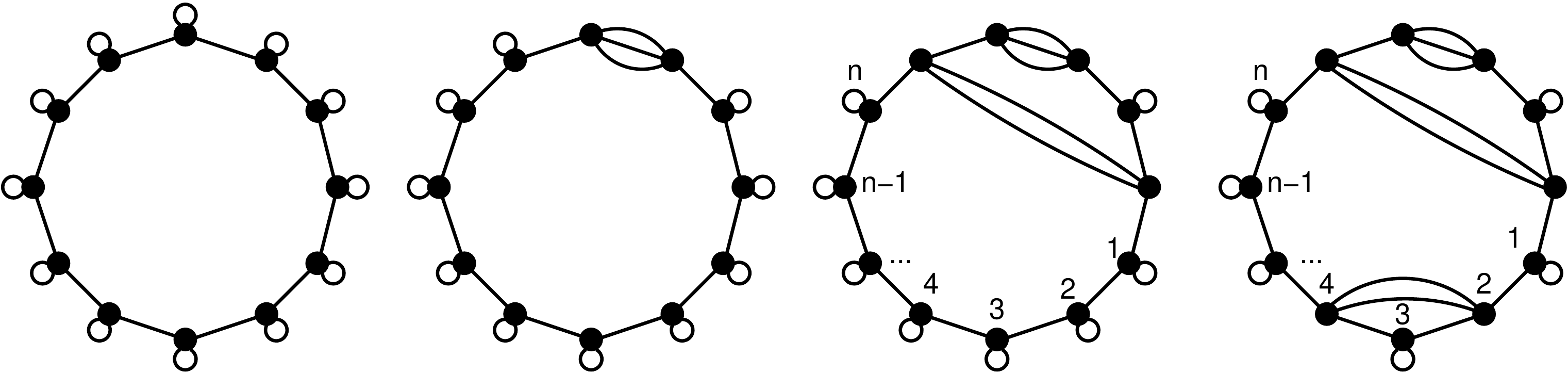}
  \end{center}
  \caption{Construction of non-isomorphic multigraphs with valency $4$ in
    the proof of Prop.~\ref{prop:injective}. The first three isomers are a
    cycle (with loops), a cycle with a single triple-bond indicating an
    ``origin'', and a graph with an additional double bond. In the third
    graph, the asymmetric arrangement of the double and triple bonds
    implies an unambiguous ordering of the remaining vertices (numbered
    from $1$ to $n$). Non-isomorphic graphs are obtained converting a
      pair of loops into a double bound.  Since each vertex has at most one
      bond in addition to the cycle, the resulting graphs correspond to
      Kleitman's ``irreducible diagrams'' \cite{Kleitman:70}. If
    crossings of bonds are excluded, the resulting induced subgraphs with
    vertex set $\{1,\dots,n\}$ are isomorphic to RNA secondary structures
    on sequences of $n$ monomers. The number $S_n$ of secondary structures
    grows asymptotically $\sim 2.6^n$ \cite{Stein:78}.  }
  \label{fig:construction}
\end{figure*}

Interestingly, the simple multigraphs in Def.~\ref{def:structform} are
sufficient to represent all conservative RNs and thus (the proper part of)
all chemical networks. Radicals and other chemical species whose structures
cannot be expressed in terms of electron pairs therefore do not add to the
universe of chemically realistic RNs.
  
Like an sf-realization, a Lewis realization does not necessarily assign
distinct multigraphs $\Gamma_x$ and $\Gamma_y$ to distinct compounds $x$
and $y$. In the case of sf-realizations, obligatory isomers must have the
same sum formula. In Lewis realizations, however, they need not have the
same multigraph. 
  \begin{proposition}
    \label{prop:injective}
    For every conservative RN $(X,\RR)$ there exists an \emph{injective
      Lewis realization} $x\mapsto \Gamma_x$.
  \end{proposition}
  \begin{proof}
    Sf-representations can be constructed to have an arbitrary number of
    atoms or moieties for each $x \in X$, that is, the vertex sets $V_x$ of
    the corresponding multigraphs $\Gamma_x$ can be chosen arbitrarily
    large. Set $\val(a)=4$ for all $a\in\mathcal{A}$ and construct an
    initial Lewis representation of compounds as cycles, as in the proof of
    Lemma~\ref{lem:suff}, but with an additional loop at each
    vertex. Consider two obligatory isomers $x\rightleftharpoons y$, and
    let the (adjacent) vertices $u,v\in V_x$ be connected (by a single
    edge).  Now replace the two loops at the corresponding vertices
    $u,v\in V_y$ by two additional edges between $u$ and $v$.  If the
    equivalence class of obligatory isomers contains more than two
    compounds, choose sets of pairs of disjoint positions along the cycles
    and replace pairs of loops by double edges. This yields circular
    matchings, familiar e.g.\ from the theory of RNA secondary structures
    \cite{Stein:78,Waterman:78}. Setting $n=|V_x|-5$, one can construct
    crossing-free circular matchings on $n$ vertices, whose number grows
    faster than $2.6^n$, see also Fig.~\ref{fig:construction}.  Thus, if
    $V_x$ is chosen large enough, an arbitrarily large set of obligatory
    isomers can be represented by non-isomorphic multigraphs. Note,
    finally, that the construction of non-isomorphic graphs does not depend
    on (the cardinality of) the atom set $\mathcal{A}$, and thus the
    construction is also applicable in the case $|\mathcal{A}|=1$, i.e.,
    $\dim\ker\SM^\trans=1$.
  \end{proof}
  The proof in particular shows that the number of vertices % $|V_x|$
  required to accommodate the obligatory isomers grows only logarithmically
  in the size of the equivalence classes of obligatory isomers.
    
\section{Discussion}
%\smallskip

\subsection{Characterization of \chemlike reaction networks}
%\smallskip

In this contribution, we have characterized reaction networks that are
\chemlike in the sense that they are consistent with the conservation of
energy and mass and allow an interpretation as transformations of chemical
molecules. It is worth noting that we arrive at our results without
invoking mass-action kinetics, which has been the focus of interest in
chemical reaction network theory since the 1970s
\cite{Horn:72,Horn:72a,Feinberg:72}.
% {Angeli:09,Chellaboina:09,Rao:13,MuellerRegensburger2014}.
Instead, we found that basic arguments from thermodynamics (without kinetic
considerations) are sufficient.  The main results of this contribution can
be summarized as follows:
\begin{itemize}
\item[(i)] A closed RN $(X,\RR)$ is thermodynamically sound if and only if
  it does not contain an irreversible futile cycle.  In particular, every
  reversible networks is thermodynamically sound.  If irreversible
  reactions are meant to proceed in a given direction for all external
  conditions (after opening the RN by adding transport reactions), then
  $(X,\RR)$ must be \emph{strictly} thermodynamically sound. Equivalently,
  a futile cycle must not contain an irreversible reaction. An
    analogous result was obtained by \cite{Gorban2011} assuming mass-action
    kinetics.
\item[(ii)] A RN $(X,\RR)$ is free of cornucopias and abysses if and only
  if it is conservative.
\item[(iii)] Both thermodynamic soundness and conservativity are completely
  determined by the stoichiometric matrix $\SM$, i.e., they are unaffected
  by catalysts.
\item[(iv)] A RN $(X,\RR)$ admits an sf-realization if and only if it is
  conservative. That is, conservative RNs admit assignments of sum formulas
  such that (i) atoms (or moieties) are conserved and (ii) two compounds
  are assigned the same sum formula if and only if they are obligatory
  isomers. Obligatory isomers, in turn, are complete determined by $\SM$.
\item[(v)] For every sf-realization of a RN $(X,\RR)$ there is also a
  Lewis-realization, i.e., an assignment of multigraphs to each compound
  such that reactions are exclusively rearrangements of edges.
\end{itemize}
Such \chemlike realizations, however, are by no means unique. In general,
the same RN has infinitely many chemical realizations corresponding to
different atomic compositions.  The structure of the stoichiometric matrix
$\SM$ of a closed RN therefore implies surprisingly little about the
underlying chemistry.

Nevertheless there is interesting information that is independent of the
concrete realization. For example, Thm.~\ref{thm:iso<->kernel} can be
reformulated as follows: The reversible completion of $(X,\RR)$ admits a
net reaction of the from $p \, x \longrightarrow q \, y$ with $x,y\in X$
and $p,q \in \N$ if and only if $q \, \mathbf{m}_{x} = p \, \mathbf{m}_{y}$
for every $\mathbf{m}\in\ker\SM^\trans$.  This identifies ``obligatory
oligomers'', necessarily composed of multiples of the same monomer.

\subsubsection*{Computational Considerations}
%\smallskip

Somewhat surprisingly, the computational problems associated with
recognizing ``\chemlike'' RNs are not particularly difficult and can be
solved by well-established methods. To see this, recall that $(X,\RR)$ is
conservative iff there is a vector $\textbf{m}\gg0$ such that
$\SM^\trans\textbf{m}=0$ and \emph{not} thermodynamically sound iff there
is a vector $\vv> 0$ such that $\SM\vv=0$ and $\vv_r> 0$ for some
$r\in\RR_{\mathrm{irr}}$ These linear programming problems can be solved in
$O((|X|+|\RR|)^{2.37})$ time \cite{Cohen:21}.

An integer (not necessarily non-negative) basis of $\ker\SM^\trans$ can be
computed exactly in polynomial time, e.g.\ using the Smith normal form, see
\cite{Newman:97}. Chubanov's algorithm finds exact rational solutions to
systems of linear equations with a strict positivity constraint. Thus is
can be employed to compute a strictly positive integer solution
$\mathbf{m}\gg0$ to $\SM^\trans\mathbf{m}=0$ in polynomial time
\cite{Chubanov:15,Roos:18}. As a consequence, an sf-realization can also be
computed explicitly in polynomial time. Each sum formula in turn can be
converted into a graph with total effort bounded by
$\max_{x\in X} \sum_{a}\mathbf{A}_{xa}\cdot |X|$, the maximal number of
atoms that appear in a sum formula times the number of molecules.

The equivalence relation $\rightleftharpoons$ for obligatory isomers is
determined by the existence of solutions to a linear equation of the form
$\SM\vv=\mathbf{w}$ and thus can also be computed in polynomial time, again
bounded by the effort for matrix multiplication for each pair $x,y \in
X$. A much more efficient approach, however, is to compute a basis of
$\ker\SM^\trans$, from which $\rightleftharpoons$ can be read off
directly. This approach easily extends to ``obligatory oligomers.''
%% obligatory isomers are introduced at the very end of the previous
%% subsection -- what is the problem here?

Treating RNs as closed systems is too restrictive to describe metabolic
networks. There, RNs are considered as open systems that allow the inflow
of nutrients and the outflow of waste products. Models of metabolism often
impose a condition of \emph{viability}. Traditionally, this is modeled as a
single export ``reaction'' $r_{bm}$ of the form
$\sum_i \alpha_i \ce{C}_i \to\varnothing$, known as the \emph{biomass
  function}~\cite{Feist:10}. It comprises all relevant precursor
metabolites $\ce{C}_i$ (forming all relevant macromolecules) in their
empirically determined proportions $\alpha_i$. Viability is then defined as
the existence of a flow $\vv>0$ with $\SM\vv=0$ and $\vv_{bm}>0$. This
linear programming problem can be tested efficiently by means of flux
balance analysis (FBA) \cite{Orth:10}. In contrast to $(X,\RR)$ being
conservative and thermodynamically sound, however, viability is a property
of the metabolic model, not of the underlying representation of the
chemistry.

\subsection{Outlook to open problems}
%\smallskip

\subsubsection*{Construction of random \chemlike networks}
%\smallskip

The formal characterization of \chemlike RNs developed here suggests
several interesting questions for further research. In particular, our
results define rather clearly how \emph{random} \chemlike RNs should be
defined and thus poses the question whether there are efficient algorithms
for their construction.  Let us consider the task of generating a random
\chemlike RN in a bit more detail. We first note that it suffices to
generate a stoichiometric matrix $\SM \in \N_0^{X\times\RR}$ that is
thermodynamically sound and conservative. If explicit catalysts are
desired, they can be added to a reaction without further restrictions.
More precisely, given $\SM$, we obtain a network with the same
stoichiometric matrix \emph{plus} catalysts by setting
\begin{equation}
  \begin{split}
    s_{xr}^- = c_{xr}, \; s_{xr}^+ = c_{xr}+s_{xr} &
    \quad\textrm{if } s_{xr}\ge 0 , \\
    s_{xr}^- = c_{xr}-s_{xr}, \;s_{xr}^+ = c_{xr} &
    \quad\textrm{if } s_{xr}\le 0 .
  \end{split}
\end{equation}
The ``catalyst matrix'' $\mathbf{C}$ may contain arbritrary integers
$c_{xr}\ge 0$. For the generation of a RN $(X,\RR)$, therefore, it can be
drawn independently of $\mathbf{S}$.

The key task of generating $(X,\RR)$ is therefore the construction of an
$|X|\times|\RR|$ integer matrix $\SM$ that is conservative and
thermodynamically sound. Both conditions amount to the (non)existence of
vectors with certain sign patterns in $\ker\SM$ and $\ker\SM^\trans$,
respectively.  In order to obtain a background model for a given chemical
RN, one might also ask for a random integer matrix that has a given left
nullspace and is thermodynamically sound.  In addition, one would probably
like to (approximately) preserve the fraction of zero entries per row and
column and the mean of the non-zero entries. % $|s_{xr}|$.
To our knowledge, no efficient exact algorithms for this problem are known.

A potentially promising alternative is the independent generation of the
complex matrix $\mathbf{Y}$ and the incidence matrix $\mathbf{Z}$ of the
complex-reaction graph. Given a fixed conservative and thermodynamically
sound RN, furthermore, one can make use of the heredity of thermodynamic
soundness and conservativity and consider random subnetworks. This approach
has been explored in particular for metabolic networks: The ensemble of
viable metabolic networks in a given chemical RN can then be sampled by a
random walk on the set of reactions \cite{MatiasRodrigues:09} or a more
sophisticated Markov-Chain-Monte-Carlo procedure \cite{Samal:10,Barve:12}.

\subsubsection*{Chemistry-like realizations} 
%\smallskip

The structural formulas constructed in Lemma~\ref{lem:cycle} are not very
``realistic' from a chemical perspective. It is of interest, therefore, if
one can construct chemically more appealing (multi-)graphs.  There seem to
be only a few constraints: (i) If a moiety $a$ appears in isolation, i.e.,
as a molecule $x=1a$, then $\val(a)$ must be even, since it contains
$\val(a)/2$ loops.  (ii) The case $\val(a)=1$ is only possible if there is
no compound composed exclusively of three or more copies of $a$ or composed
of more than two moieties with valency $1$. (iii) It is well known that the
sum of degrees must be even for every multigraph, and connectedness implies
$\sum_u \val(u)\ge 2(|V|-1)$ \cite{Edmonds:64}.
  
The problem of finding multigraph realizations is closely related to, but
not the same as, the problem determining the realizability of degree
sequences in graphs \cite{Meierling:09} or multigraphs \cite{Sierksma:91}.
As in graph theory, it seems to be of particular interest to study
realizability by structural formula in the presence of additional
constraints on admissible graphs.

An advantage of considering the multigraphs specified in
Def.~\ref{def:structform} instead of the full range of Lewis structures is
that a well-established mathematical theory is available. However,
``multigraphs with semi-edges'', which are essentially equivalent to Lewis
structures, have been studied occasionally in recent years
\cite{Getzler:98,Mednykh:15} and may be an appealing framework, in
particular, when restricted realizations are considered.

\subsubsection*{Infinite RNs}
%\smallskip

Throughout this contribution, we have assumed that $(X,\RR)$ is finite. In
general, however, chemical universes are infinite, at least in
principle. The simplest example of infinite families are polymers. It is of
interest, therefore, to develop a theory of infinite reaction networks. To
this end, one could follow e.g.\ \cite{OstermeierL:12a}, where also
infinite directed hypergraphs are considered, and further extend the
literature on countably infinite undirected hypergraphs, see e.g.\
\cite{Banakh:19,Bustamante:20} and the references therein. Most previous
work pre-supposed $k$-uniformity, i.e., hyper-edges of (small) finite
cardinality, matching well with the situation in chemical RNs. Every sub-RN
of an infinite RN induced by a finite vertex set $Y\subset X$ can be
assumed to support only a finite number of reactions (directed hyperedges)
$\RR_Y\subset\RR$. This amounts to assuming that a sub-RN induced by finite
set of compounds $Y$ is a finite RN. Every finite sub-RN of a ``\chemlike''
infinite RN, furthermore, needs to be conservative and thermodynamically
sound. Infinite RNs will not be locally finite, in general, since every
compound $x\in X$ may have infinitely many reaction partners, e.g., all
members of a polymer family. Thus $x$ may appear in an infinite number of
reactions. These simple observations suggest infinite ``\chemlike'' RNs are
non-trivial structures whose study may turn out to be a worth-while
mathematical endeavor.
  
%%%%%%%%%%%%%%%%%%%%%%%%%%%%%%%%%%%%%%%%%%%%%%
%%                                          %%
%% Backmatter begins here                   %%
%%                                          %%
%%%%%%%%%%%%%%%%%%%%%%%%%%%%%%%%%%%%%%%%%%%%%%

%\begin{backmatter}

%\section{Abbreviations}  

\section*{Competing interests}
The authors declare that they have no competing interests.

\section*{Author contributions}
CF and PFS designed the study, SM and PFS proved the mathematical results,
all authors contributed to the interpretation of the results and the
writing of the manuscript.

\section*{Acknowledgements}

This research was funded in part by the German Federal Ministry of
Education and Research within the project Center for Scalable Data
Analytics and Artificial Intelligence (ScaDS.AI.) Dresden/Leipzig (BMBF
01IS18026B).  SM was supported by the Austrian Science Fund (FWF), project
P33218.

\footnotesize

\bibliographystyle{vancouver}
\bibliography{randnet}

%\end{backmatter}

\end{document}